\PassOptionsToPackage{cal=euler}{mathalfa}

\documentclass[
    subscriptcorrection,
    upint,
    varvw,
    barcolor=black,
    balance,
    hyphenate,
    french,
    pdf-a,
    nolists,
    nofoot
]{asmejour}

\hypersetup{%
	pdfauthor={Mihails Milehins},
	pdftitle={template.pdf},
	pdfkeywords={Bouc-Wen Model, Collision, Hysteresis, Impact, Rigid Body Dynamics}
}

\usepackage{paralist}
\usepackage{tikz}
\usetikzlibrary{calc}
\usetikzlibrary{arrows.meta}
\usepackage{pgfplots}
\usepackage{physics}
\usepackage{amsthm}
\usepackage{makecell}

\newtheorem{theorem}{Theorem}[section]
\newtheorem{lem}[theorem]{Lemma}
\newtheorem{prop}[theorem]{Proposition}

\theoremstyle{definition}
\newtheorem{definition}{Definition}[section]
\theoremstyle{remark}

\JourName{Computational and Nonlinear Dynamics}

\begin{document}



\SetAuthorBlock{Mihails Milehins\CorrespondingAuthor}{%
	Department of Mechanical Engineering,\\
	Auburn University,\\
	Auburn, AL 36849,\\
    email: mzm0390@auburn.edu} 

\SetAuthorBlock{Dan B. Marghitu}{
	Department of Mechanical Engineering,\\
	Auburn University,\\
	Auburn, AL 36849,\\
    email: marghdb@auburn.edu}

\title{Incremental Collision Laws Based on the Bouc-Wen Model: Improved Collision Models and Further Results}


\begin{abstract}
In the article titled ``The Bouc-Wen Model for Binary Direct Collinear Collisions of Convex Viscoplastic Bodies'' and published in the Journal of Computational and Nonlinear Dynamics (Volume 20, Issue 6, June 2025), the authors studied mathematical models of binary direct collinear collisions of convex viscoplastic bodies that employed two incremental collision laws based on the Bouc-Wen differential model of hysteresis. It was shown that the models possess favorable analytical properties, and several model parameter identification studies were conducted, demonstrating that the models can accurately capture the nature of a variety of collision phenomena. In this article, the aforementioned models are augmented by modeling the effects of external forces as time-dependent inputs. Furthermore, the range of the parameters under which the models possess favorable analytical properties is extended to several corner cases that were not considered in the prior publication. Finally, the previously conducted model parameter identification studies are extended, and an additional model parameter identification study is provided in an attempt to validate the ability of the augmented models to represent the effects of external forces.
\end{abstract}

\date{\today} 

\maketitle 

\section{Introduction}\label{sec:introduction}

There exist two primary approaches for modeling of systems of rigid bodies with contacts: nonsmooth dynamics formulations (e.g., see Refs. \cite{panagiotopoulos_inequality_1985, moreau_unilateral_1988, pfeiffer_multibody_2004, stewart_dynamics_2011, goebel_hybrid_2012, brogliato_nonsmooth_2016, sanfelice_hybrid_2021}) and continuous formulations (e.g., see Refs. \cite{terzopoulos_elastically_1987, platt_constraint_1988, moore_collision_1988, movahedi-lankarani_canonical_1988}). This article is concerned with continuous formulations, which require a continuous dynamic model that can describe the evolution of the contact force during the collision events (e.g., see Refs. \cite{machado_compliant_2012, corral_nonlinear_2021}). Such dynamic models are referred to as incremental collision laws.

In Ref. \cite{milehins_boucwen_2025}, the authors studied mathematical models of binary direct collinear collisions of convex viscoplastic bodies using two incremental collision laws based on the Bouc-Wen differential model of hysteresis (\cite{bouc_forced_1968, bouc_modemathematique_1971, wen_method_1976}, see also Ref. \cite{ikhouane_systems_2007}).\footnote{The specific form of the Bouc-Wen model that was used in Ref. \cite{milehins_boucwen_2025} and in this work is based on the form employed in Ref. \cite{ma_parameter_2004}.} These collision laws are the Bouc-Wen-Simon-Hunt-Crossley Collision Law (BWSHCCL), an extension of the Simon-Hunt-Crossley Collision Law (see \cite[Simon (1967), as cited in Ref.][]{brogliato_nonsmooth_2016} and Ref. \cite{hunt_coefficient_1975}) that is formed by a parallel connection of a nonlinear viscous energy dissipation element and a Bouc-Wen hysteretic element with a nonlinear output function, and the Bouc-Wen-Maxwell Collision Law (BWMCL), an extension of the Maxwell Collision Law (see Refs. \cite{maxwell_dynamical_1867, johnson_contact_1985, butcher_characterizing_2000}) that is formed by a series connection of a linear viscous energy dissipation element and a Bouc-Wen hysteretic element with a nonlinear output function. The BWSHCCL was stated as\footnote{It is assumed that a consistent system of units is used for all dimensional quantities (the units are often omitted). For mathematical conventions see Appendix \ref{sec:NC}.}
\begin{equation}\label{eq:BWSHCCL}
\begin{cases}
\dot{x} = u\\
\dot{z} = A u - \beta \abs{z}^{n - 1} z \abs{u} - \gamma \abs{z}^n u \\
F = - \alpha k \abs{x}^{p - 1} x - \alpha_c k \abs{z}^{p - 1} z - c \abs{x}^p u
\end{cases}
\end{equation}
where $x \in \mathbb{R}$ is a state variable that represents the relative displacement of the centers of mass of the colliding bodies relative to their initial relative displacement (i.e., the relative displacement at the time of the collision) or, equivalently, the relative displacement at the contact interface (in what follows, relative displacement at a contact interface will be referred to simply as relative displacement), $z \in \mathbb{R}$ is a state variable that represents the hysteretic displacement associated with the Bouc-Wen model, $u \in \mathbb{R}$ is an input variable that represents the relative velocity of the centers of mass of the colliding bodies or, equivalently, the relative velocity at the contact interface (in what follows, relative velocity at a contact interface will be referred to simply as relative velocity), $F \in \mathbb{R}$ is an output variable that represents the contact force; the model is parameterized by $A, k \in \mathbb{R}_{>0}$, $\alpha \in (0, 1)$, $c, \beta \in \mathbb{R}_{\geq 0}$, $\gamma \in [- \beta, \beta]$, and $n, p \in \mathbb{R}_{\geq 1}$, with $\alpha_c \triangleq 1 - \alpha$.\footnote{In what follows, $\alpha_c$ will always be used as an abbreviation for $1 - \alpha$ (without an explicit elaboration).} 

The BWMCL was stated as
\begin{equation}\label{eq:BWMCL}
\begin{cases}
\dot{r} = \alpha \frac{k}{c} \abs{y}^{p - 1} y + \alpha_c \frac{k}{c} \abs{z}^{p - 1} z\\
\dot{y} = - \dot{r} + u\\
\dot{z} = A \dot{y} - \beta \abs{z}^{n - 1} z \abs{\dot{y}} - \gamma \abs{z}^n \dot{y} \\
F = - c \dot{r} = - \alpha k \abs{y}^{p - 1} y - \alpha_c k \abs{z}^{p - 1} z
\end{cases}
\end{equation}
where $r \in \mathbb{R}$ is a state variable that represents the relative displacement of a linear viscous energy dissipation element, $y \in \mathbb{R}$ is a state variable that represents the relative displacement of the Bouc-Wen hysteretic element, $z \in \mathbb{R}$ is a state variable that represents the hysteretic displacement in the Bouc-Wen hysteretic element, $u \in \mathbb{R}$ is an input variable that represents the relative velocity, $F \in \mathbb{R}$ is an output variable that represents the contact force; the model is parameterized by $A, k, c \in \mathbb{R}_{>0}$, $\alpha \in (0, 1)$, $\beta \in \mathbb{R}_{\geq 0}$, $\gamma \in [- \beta, \beta]$, and $n, p \in \mathbb{R}_{\geq 1}$.

The Bouc-Wen-Simon-Hunt-Crossley Collision Model (BWSHCCM), which is meant to represent binary direct collinear collisions and employs the BWSHCCL to model the contact force, was stated as
\begin{equation}\label{eq:BWSHCCM}
\begin{cases}
\dot{x} = v \\
\dot{z} = A v - \beta \abs{z}^{n - 1} z \abs{v} - \gamma \abs{z}^n v \\
\dot{v} = - \alpha \frac{k}{m} \abs{x}^{p - 1} x - \alpha_c \frac{k}{m} \abs{z}^{p - 1} z - \frac{c}{m} \abs{x}^p v \\
\begin{matrix} x(0) = 0, & z(0) = 0, & v(0) = -v_0 \end{matrix}
\end{cases}
\end{equation}
where $x \in \mathbb{R}$ is a state variable that represents the relative displacement, $z \in \mathbb{R}$ is a state variable that represents the hysteretic displacement associated with the BWSHCCL, $v \in \mathbb{R}$ is a state variable that represents the relative velocity, $m \in \mathbb{R}_{>0}$ is a parameter that represents the effective mass of the colliding bodies (an explanation is provided in Sec. \ref{sec:mps}; see also Ref. \cite{nikravesh_determination_2023}), $v_0 \in \mathbb{R}_{>0}$ is a parameter that describes the initial relative velocity; other parameters are adopted from the BWSHCCL.

The Bouc-Wen-Maxwell Collision Model (BWMCM), which employs the BWMCL to model the contact force, was stated as
\begin{equation}\label{eq:BWMCM}
\begin{cases}
\dot{r} = \alpha \frac{k}{c} \abs{y}^{p - 1} y + \alpha_c \frac{k}{c} \abs{z}^{p - 1} z\\
\dot{y} = w \\
\dot{z} = A w - \beta \abs{z}^{n - 1} z \abs{w} - \gamma \abs{z}^n w \\
\dot{w} = - \frac{c}{m} \dot{r} - \alpha  p \frac{k}{c} \abs{y}^{p - 1} \dot{y} - \alpha_c p \frac{k}{c} \abs{z}^{p - 1} \dot{z} \\
\begin{matrix} r(0) = y(0) = z(0) = 0, & w(0) = -v_0 \end{matrix}
\end{cases}
\end{equation}
where $r \in \mathbb{R}$ is a state variable that represents the relative displacement of the linear viscous energy dissipation element associated with the BWMCL, $y \in \mathbb{R}$ is a state variable that represents the relative displacement of the Bouc-Wen hysteretic element associated with the BWMCL, $z \in \mathbb{R}$ is a state variable that represents the hysteretic displacement of the Bouc-Wen hysteretic element associated with the BWMCL, $w \in \mathbb{R}$ is a state variable that represents the relative velocity of the Bouc-Wen hysteretic element associated with the BWMCL, $m \in \mathbb{R}_{>0}$ is a parameter that represents the effective mass of the colliding bodies (an explanation is provided in Sec. \ref{sec:mps}; see also Ref. \cite{nikravesh_determination_2023}), $v_0 \in \mathbb{R}_{>0}$ is a parameter that describes the initial relative velocity; other parameters are adopted from the BWMCL. The relative displacement $x \in \mathbb{R}$ and the relative velocity $v \in \mathbb{R}$ can be recovered by augmenting the BWMCM with the output function given by
\begin{equation}\label{eq:BWMCM_output}
(r, y, z, w) \mapsto (r, y, z, w, r + y, \dot{r} + \dot{y}) \triangleq (r, y, z, w, x, v)
\end{equation}

The nondimensionalized form of the BWSHCCM, referred to as the Nondimensionalized Bouc-Wen-Simon-Hunt-Crossley Collision Model (NDBWSHCCM), was given by 
\begin{equation}\label{eq:BWSHCCM_nd}
\begin{cases}
\dot{X} = V \\
\dot{Z} = V - B \abs{Z}^{n - 1} Z \abs{V} - \mathit{\Gamma} \abs{Z}^n V \\
\dot{V} = - \kappa \abs{X}^{p - 1} X - \kappa_c \abs{Z}^{p - 1} Z - \sigma \abs{X}^p V \\
\begin{matrix} X(0) = 0, & Z(0) = 0, & V(0) = -1 \end{matrix}
\end{cases}
\end{equation}
The relationships between the nondimensionalized and dimensional variables are given by $T \triangleq t/T_c$, $X \triangleq x/X_c$, $Z \triangleq z/Z_c$, $V \triangleq v/(X_c/T_c)$. The parameters that were used for the nondimensionalization are given in Table \ref{tab:ND}; as previously, $\kappa_c \triangleq 1 - \kappa$.\footnote{In what follows, $\kappa_c$ will always be used as an abbreviation for $1 - \kappa$ (without an explicit elaboration).} 

The nondimensionalized form of the BWMCM, referred to as the Nondimensionalized Bouc-Wen-Maxwell Collision Model (NDBWMCM), was given by 
\begin{equation}\label{eq:BWMCM_nd}
\begin{cases}
\dot{R} = \kappa \sigma \abs{Y}^{p - 1} Y + \kappa_c \sigma \abs{Z}^{p - 1} Z\\
\dot{Y} = W \\
\dot{Z} = W - B \abs{Z}^{n - 1} Z \abs{W} - \mathit{\Gamma} \abs{Z}^n W \\
\dot{W} = - \frac{1}{\sigma} \dot{R} - \kappa p \sigma \abs{Y}^{p - 1} \dot{Y} - \kappa_c p \sigma \abs{Z}^{p - 1} \dot{Z} \\
\begin{matrix} R(0) = Y(0) = Z(0) = 0, & W(0) = -1 \end{matrix}
\end{cases}
\end{equation}
with the output function given by 
\begin{equation}\label{eq:BWMCM_nd_output}
(R, Y, Z, W) \mapsto (R, Y, Z, W, R + Y, \dot{R} + \dot{Y}) \triangleq (R, Y, Z, W, X, V)
\end{equation}
The relationships between the nondimensionalized and dimensional variables are given by $T \triangleq t/T_c$, $R \triangleq r/X_c$, $Y \triangleq y/X_c$, $Z \triangleq z/Z_c$, $W \triangleq w/(X_c/T_c)$, $X \triangleq x/X_c$, $V \triangleq v/(X_c/T_c)$. The parameters that were used for nondimensionalization are given in Table \ref{tab:ND}. 

In Ref. \cite{milehins_boucwen_2025}, the authors show that if the NDBWSHCCM is parameterized by $B \in \mathbb{R}_{\geq 0}$, $\mathit{\Gamma} \in [-B, B]$, $\kappa \in (0, 1)$, $\sigma \in \mathbb{R}_{\geq 0}$, $n, p \in \mathbb{R}_{\geq 1}$, then the NDBWSHCCM has a unique bounded solution on any time interval $[0, T_e)$ with $T_e \in \mathbb{R}_{>0} \cup \{ +\infty \}$. The authors also show that if the NDBWMCM is parameterized by $B \in \mathbb{R}_{> 0}$, $\mathit{\Gamma} \in (-B, B)$, $\kappa \in (0, 1)$, $\sigma \in \mathbb{R}_{>0}$, $n \in \mathbb{R}_{\geq 1}$, $p \in \mathbb{R}_{\geq 2} \cup \{ 1 \}$, then the NDBWMCM has a unique bounded solution on any time interval $[0, T_e)$ with $T_e \in \mathbb{R}_{>0} \cup \{ +\infty \}$. Moreover, the output associated with this solution is bounded. Furthermore, the authors show that (under a slightly more restricted set of parameters) the solutions of the NDBWSHCCM and the NDBWMCM converge to an infinite set of equilibrium points at a finite distance from the origin (see also Refs. \cite{jayawardhana_stability_2012, ouyang_absolute_2014}). Lastly, the authors conduct two model parameter identification studies that demonstrate that both the NDBWSHCCM and the NDBWMCM can accurately represent a variety of collision phenomena. 

While Ref. \cite{milehins_boucwen_2025} offers significant contributions to the analysis and validation of the NDBWSHCCM and the NDBWMCM, the models and the associated analytical framework can be improved. The goal of the present study is to offer a natural extension of the work presented in Ref. \cite{milehins_boucwen_2025}.

\begin{table*}[t]
\caption{Parameters for nondimensionalization of the BWSHCCM and the BWMCM}\label{tab:ND}
\centering{
\begin{tabular*}{0.8\textwidth}{@{\hspace*{1.5em}}@{\extracolsep{\fill}}ccc@{\hspace*{1.5em}}}
\toprule
\multicolumn{1}{c}{Parameters} & \multicolumn{1}{c}{BWSHCCM} & \multicolumn{1}{c}{BWMCM} \\ 
\midrule
$T_c$             & $\left( \frac{1}{\alpha + \alpha_c A^p} \right)^{\frac{1}{p + 1}} \left( \frac{m}{k} \right)^{\frac{1}{p + 1}} v_0^{-\frac{p - 1}{p + 1}}$ & $\left(\frac{1}{\alpha + \alpha_c A^p}\right)^{\frac{1}{p + 1}} \left(\frac{m}{k}\right)^{\frac{1}{p + 1}} v_0^{-\frac{p - 1}{p + 1}}$ \\
$X_c$  & $\left( \frac{1}{\alpha + \alpha_c A^p} \right)^{\frac{1}{p + 1}} \left( \frac{m}{k} \right)^{\frac{1}{p + 1}} v_0^{\frac{2}{p + 1}}$ & $\left(\frac{1}{\alpha + \alpha_c A^p}\right)^{\frac{1}{p+1}} \left(\frac{m}{k}\right)^{\frac{1}{p + 1}} v_0^{\frac{2}{p+1}}$   \\
$Z_c$          & $\left( \frac{1}{\alpha + \alpha_c A^p} \right)^{\frac{1}{p + 1}} A \left( \frac{m}{k} \right)^{\frac{1}{p + 1}} v_0^{\frac{2}{p + 1}}$   & $\left(\frac{1}{\alpha + \alpha_c A^p}\right)^{\frac{1}{p+1}} A \left(\frac{m}{k}\right)^{\frac{1}{p + 1}} v_0^{\frac{2}{p + 1}}$ \\
$B$         & $\left( \frac{A^{p + 1}}{\alpha + \alpha_c A^p} \right)^{\frac{n}{p + 1}} \frac{\beta}{A} \left( \frac{m}{k} \right)^{\frac{n}{p + 1}} v_0^{\frac{2 n}{p + 1}}$ & $\left(\frac{A^{p+1}}{\alpha + \alpha_c A^p}\right)^{\frac{n}{p+1}} \frac{\beta}{A} \left(\frac{m}{k}\right)^{\frac{n}{p + 1}} v_0^{\frac{2 n}{p + 1}}$ \\
$\mathit{\Gamma}$ & $\left( \frac{A^{p + 1}}{\alpha + \alpha_c A^p} \right)^{\frac{n}{p + 1}} \frac{\gamma}{A} \left( \frac{m}{k} \right)^{\frac{n}{p + 1}} v_0^{\frac{2 n}{p + 1}}$    & $\left(\frac{A^{p + 1}}{\alpha + \alpha_c A^p}\right)^{\frac{n}{p + 1}} \frac{\gamma}{A} \left(\frac{m}{k}\right)^{\frac{n}{p + 1}} v_0^{\frac{2 n}{p + 1}}$ \\
$\kappa$               & $\frac{\alpha}{\alpha + \alpha_c A^p}$    & $\frac{\alpha }{\alpha + \alpha_c A^p}$ \\
$\sigma$                & $\frac{1}{\alpha + \alpha_c A^p} \frac{c}{k} v_0$    & $\left(\alpha + \alpha_c A^p\right)^{\frac{1}{p + 1}} \frac{1}{c} \left( m^p k \right)^{\frac{1}{p+1}} v_0^{\frac{p - 1}{p + 1}}$ \\
\bottomrule
\end{tabular*}
}
\end{table*}

\section{Contributions and Outline}\label{sec:contributions}

The following list identifies several possible avenues for improvement of the study presented in Ref. \cite{milehins_boucwen_2025}:
\begin{compactitem}
\item Both the BWSHCCM and the BWMCM were designed under the assumption that the only force that acts on the bodies during the collision is the contact force. However, sometimes, external forces that act on the bodies while the bodies maintain contact cannot be ignored (e.g., see Refs. \cite{tatara_effects_1977, falcon_behavior_1998, quinn_finite_2004, sorace_high_2009, ye_size-dependent_2017, shen_contact_2018, xiang_comparative_2018, carvalho_exact_2019, yardeny_experimental_2020, villegas_impact_2021, chatterjee_approximate_2022, bartz_gravity_2023, akhan_low_2024, shen_contact_2024}).
\item The analysis of the NDBWMCM was not performed for the following choices of parameters: $B = 0$, $\mathit{\Gamma} \in \{ -B, B \}$, $p \in (1, 2)$. These parameters lie within the physically plausible range and may be important for applications.
\item The parameter identification study based on the dataset in Fig. 9.5 in Ref. \cite{cross_physics_2011} was restricted to the BWSHCCM.
\end{compactitem}
The goal of the present article is to address the issues that were outlined in the list above. The BWSHCCM and the BWMCM will be augmented by modeling external forces as inputs that belong to certain function spaces, the analysis of the models will be revised to include the corner cases that were described in the list above, the model parameter identification studies will be updated, and a further model parameter identification study will be provided to validate the BWSHCCM and the BWMCM augmented with the action of external forces. The remainder of the article is organized as follows:
\begin{compactitem}
\item Section \ref{sec:mps} introduces a high-level model of the physical system.
\item Section \ref{sec:BWSHCCM} presents an augmented form of the BWSHCCM that includes the effects of external forces.
\item Section \ref{sec:BWMCM} presents an augmented form of the BWMCM that includes the effects of external forces and is more convenient for analysis in comparison to the form of the model presented in Ref. \cite{milehins_boucwen_2025}.
\item Section \ref{sec:AE} provides an application example.
\item Section \ref{sec:MPI} provides an improvement of the methodology for the identification of the parameters of the collision models and presents several applications of the methodology.
\item Section \ref{sec:conclusions} provides conclusions and recommendations.
\item Appendices \ref{sec:NC}-\ref{sec:SDA} describe the mathematical conventions, provide proofs of the main results presented in Sec. \ref{sec:BWSHCCM} and Sec. \ref{sec:BWMCM}, and describe the methodology that was used for numerical simulation in Sec. \ref{sec:AE} and Sec. \ref{sec:MPI}.
\end{compactitem}

\section{Model of the Physical System}\label{sec:mps}

The discussion that follows is with reference to Fig. \ref{fig:nc}. The notational conventions for mechanics are adopted from Ref. \cite{roithmayr_dynamics_2016} and Ref. \cite{stronge_impact_2018}. It is assumed that $\mathcal{B}_1$ is a compact and strictly convex rigid body and $\mathcal{B}_2$ is a convex rigid body with a topologically smooth surface. The bodies are assumed to come into contact at the time $t_0 = 0 \in \mathbb{R}_{\geq 0}$ with their centers of mass $G_1 = (l_1, 0, 0)$ and $G_2 = (-l_2, 0, 0)$ such that $l_1, l_2 \in \mathbb{R}_{\geq 0}$ lying on the line $A'B'$ that passes through the point of contact $C \triangleq (0, 0, 0)$.\footnote{If needed, the collision detection may be performed according to the methodology suggested in Ref. \cite{pfeiffer_multibody_2004}; the geometric constraints imposed on the bodies ensure that there is at most one point of contact.} The velocity fields of both bodies are assumed to be uniform and parallel to this line. The configuration, as hereinbefore described, corresponds to a binary direct collinear impact (e.g., see Ref. \cite{stronge_impact_2018}). Building upon the methodology proposed in Ref. \cite{stronge_impact_2018}, it shall be assumed that while the bodies remain in contact, the motion of the system is governed by the laws of rigid body dynamics (Newton \cite{newton_mathematical_1729}), with the contact point described as an infinitesimal deformable particle \cite{stronge_impact_2018}.\footnote{It should be noted that similar modeling frameworks have been used since the eighteenth century (e.g., see Ref. \cite{euler_force_1746}).} The contact force $\mathbf{F} \triangleq F \hat{\mathbf{n}}_1$ with $F : \mathbb{R}_{\geq 0} \longrightarrow \mathbb{R}$ being continuous is applied to $\mathcal{B}_1$ at the point $C_1$, which is a point of $\mathcal{B}_1$ that coincides with $C$ at the time of contact. Then, the corresponding force that acts on the body $\mathcal{B}_2$ will be $-\mathbf{F} \triangleq -F \hat{\mathbf{n}}_1$, applied at the point $C_2$ coincident with the point $C$ at the time of contact. It is also assumed that the resultant of the external forces that act on the body $\mathcal{B}_i$\footnote{In what follows, it shall always be assumed that $i$ ranges over the set $\{ 1, 2 \}$.} is $\mathbf{u}_i \triangleq u_i \hat{\mathbf{n}}_1$ with $u_i : \mathbb{R}_{\geq 0} \longrightarrow \mathbb{R}$ being continuous. The point of application of $\mathbf{u}_i$ (with reference to $\mathcal{B}_i$) will be $E_i \triangleq (x_i^e, 0, 0)$ with $x_i^e \in \mathbb{R}$. Depending on the context, it will be assumed that $u_i$ belongs either to $\mathcal{U}_1$, the space of all continuous functions with the domain $\mathbb{R}_{\geq 0}$ and the codomain $\mathbb{R}$ such that $\lVert u \rVert_1 \triangleq \int_0^{+\infty} \abs{u(s)} ds < +\infty$ for all $u \in \mathcal{U}_1$, or $\mathcal{U}_{\infty}$, the space of all continuous bounded functions with the domain $\mathbb{R}_{\geq 0}$ and the codomain $\mathbb{R}$. $\mathcal{U}$ will be used to denote either $\mathcal{U}_1$ or $\mathcal{U}_{\infty}$.

Given the aforementioned assumptions, the motion of the system will be confined to a single dimension, and the two coordinates $x_{m,1} \in \mathbb{R}$ and $x_{m,2} \in \mathbb{R}$ such that $\mathbf{r}_{G_1/C} \triangleq x_{m,1} \hat{\mathbf{n}}_1$ and $\mathbf{r}_{G_2/C} \triangleq x_{m,2} \hat{\mathbf{n}}_1$ suffice to describe the dynamics of the system. More specifically, the motion can be described by the following Initial Value Problem (IVP):
\begin{equation}\label{eq:primary}
\begin{cases}
\ddot{x}_{m,1} = m_1^{-1} F + m_1^{-1} u_1 & x_{m,1}(0) = l_1, \dot{x}_{m,1}(0) = v_{m,1,0} \\
\ddot{x}_{m,2} = -m_2^{-1} F + m_2^{-1} u_2 & x_{m,2}(0) = -l_2, \dot{x}_{m,2}(0) = v_{m,2,0} 
\end{cases}
\end{equation}
where $l_1, l_2 \in \mathbb{R}_{\geq 0}$ and $v_{m,1,0}, v_{m,2,0} \in \mathbb{R}$ are parameters such that $-(v_{m,1,0} - v_{m,2,0}) \in \mathbb{R}_{>0}$.

Defining $x_1 \triangleq x_{m,1} - l_1$, $x_2 \triangleq x_{m,2} + l_2$, $v_{1, 0} \triangleq v_{m,1,0}$, and $v_{2, 0} \triangleq v_{m,2,0}$, the IVP given by Eq. \eqref{eq:primary} can be restated as 
\begin{equation}\label{eq:primary_contact}
\begin{cases}
\ddot{x}_1 = m_1^{-1} F + m_1^{-1} u_1 & x_1(0) = 0, \dot{x}_1(0) = v_{1,0} \\
\ddot{x}_2 = -m_2^{-1} F + m_2^{-1} u_2 & x_2(0) = 0, \dot{x}_2(0) = v_{2,0} 
\end{cases}
\end{equation}
This IVP describes the evolution of the contact point of each body. 

Denoting 
\begin{equation}
m \triangleq \frac{m_1 m_2}{m_1 + m_2}
\end{equation}
\begin{equation} 
x \triangleq x_1 - x_2
\end{equation}
\begin{equation} 
v \triangleq \dot{x} = \dot{x}_1 - \dot{x}_2
\end{equation}
\begin{equation}
v_0 \triangleq -(v_{1,0} - v_{2,0})
\end{equation}
\begin{equation}\label{eq:u}
u \triangleq \frac{m_2 u_1 - m_1 u_2}{m_1 + m_2}
\end{equation}
the equations of motion can be transformed to
\begin{equation}\label{eq:main}
\begin{cases}
\dot{x} = v & x(0) = 0 \\
\dot{v} = m^{-1} F + m^{-1} u & v(0) = -v_0 \\
\end{cases}
\end{equation}
It should be noted that $m \in \mathbb{R}_{>0}$ describes the effective mass of the colliding bodies (e.g., see Ref. \cite{nikravesh_determination_2023}). As previously, $x$ will be referred to as the relative displacement and $v$ as the relative velocity. It should be noted that if (by abuse of notation) $m_2 = +\infty$, then $m^{-1} = m_1^{-1}$ and $u = u_1$. This situation corresponds to the collision of a body $\mathcal{B}_1$ of finite mass with a stationary body $\mathcal{B}_2$.

Assuming (global) existence and uniqueness of solutions of the IVP given by Eq. \eqref{eq:main} on a non-degenerate time interval $I \subseteq \mathbb{R}_{\geq 0}$ with $0 \in I$, the time of the separation $t_s \in \mathbb{R}_{> 0} \cup \{ + \infty \}$ is defined as
\begin{equation}\label{eq:main_t_s}
t_s \triangleq \inf \{ t \in I_{\geq 0} : F(t) \leq 0 \leq v(t) \}
\end{equation}
for any given solution. The duration of the collision will be denoted $t_d \in \mathbb{R}_{\geq 0}$ and defined as $t_d \triangleq t_s - t_0$. Therefore, in the context of this study, $t_d = t_s$. Under the same assumptions, the (kinematic) Coefficient of Restitution (CoR) $e \in \mathbb{R}$ is given by\footnote{See Ref. \cite{stronge_impact_2018} for a conceptual description of the kinematic coefficient of restitution, which is usually attributed to Sir Isaac Newton \cite{newton_mathematical_1729}.}
\begin{equation}\label{eq:main_cor}
e \triangleq 
\begin{cases}
-v(t_s)/v(0) & t_s \neq +\infty \\
0 & t_s = +\infty
\end{cases}
\end{equation}

For some applications, it may also be of interest to understand the evolution of the location of the center of mass of the entire system. First, introduce the parameters $\eta, \eta_c \in (0, 1)$ given by 
\begin{equation}
\eta \triangleq \frac{m_1}{m_1 + m_2}
\end{equation}
\begin{equation}
\eta_c \triangleq 1 - \eta = \frac{m_2}{m_1 + m_2}
\end{equation}
The location of the center of mass is given by 
\begin{equation}
\mathbf{r}_{G/C} = \eta \mathbf{r}_{G_1/C} + \eta_c \mathbf{r}_{G_2/C} \triangleq x_m \hat{\mathbf{n}}_1
\end{equation}
Then,
\begin{equation}\label{eq:x_m}
x_m \triangleq \eta x_{m,1} + \eta_c x_{m,2} = \eta x_1 + \eta_c x_2 + \eta l_1 - \eta_c l_2
\end{equation}
Introducing $v_m \triangleq \dot{x}_m$, the IVP associated with the evolution of the location of the center of mass of the system is given by
\begin{equation}\label{eq:com}
\begin{cases}
\dot{x}_m = v_m & x_m (0) = x_{m,0} \\
\dot{v}_m = (m_1 + m_2)^{-1} (u_1 + u_2) & v_m (0) = v_{m,0}
\end{cases}
\end{equation}
where $x_{m,0} \in \mathbb{R}$ and $v_{m,0} \in \mathbb{R}$ are given by
\begin{equation}\label{eq:x_m_0}
x_{m,0} \triangleq \eta l_1 - \eta_c l_2
\end{equation}
\begin{equation}\label{eq:v_m_0}
v_{m,0} = \eta v_{m,1,0} + \eta_c v_{m,2,0}
\end{equation}
respectively. Once $x$ and $x_m$ are known, the evolution of the locations of the centers of mass of the individual bodies can be recovered via
\begin{equation}\label{eq:x_m_1}
x_{m,1} = x_m + \eta_c x + \eta_c (l_1 + l_2)\\
\end{equation}
\begin{equation}\label{eq:x_m_2}
x_{m,2} = x_m - \eta x - \eta (l_1 + l_2)\\
\end{equation}
Then, the velocities of the centers of mass of the bodies can be obtained via 
\begin{equation}\label{eq:v_m_1}
v_{m,1} \triangleq \dot{x}_{m,1} = v_m + \eta_c v
\end{equation}
\begin{equation}\label{eq:v_m_2}
v_{m,2} \triangleq \dot{x}_{m,2} = v_m - \eta v
\end{equation}

Introduction of the variables $T \triangleq t/T_c$, $X \triangleq x/X_c$, $X_m \triangleq (x_m - x_{m,0})/X_c$, $X_{m,1} \triangleq (x_{m,1} - l_1)/X_c$, $X_{m,2} \triangleq (x_{m,2} + l_2)/X_c$, $V \triangleq v/(X_c/T_c)$, $V_m \triangleq v_m/(X_c/T_c)$, $V_{m,1} \triangleq v_{m, 1}/(X_c/T_c)$, $V_{m,2} \triangleq v_{m,2}/(X_c/T_c)$ using the parameters listed in Table \ref{tab:ND} and nondimensionalization of Eq. \eqref{eq:com} using the methodology presented in Ref. \cite{logan_applied_2013} results in the system
\begin{equation}\label{eq:COM}
\begin{cases}
\dot{X}_m = V_m & X_m (0) = 0 \\
\dot{V}_m = U_1 + U_2 & V_m (0) = V_{m,0}
\end{cases}
\end{equation}
where $V_{m,0} \in \mathbb{R}$ and $U_1, U_2 : \mathbb{R}_{\geq 0} \longrightarrow \mathbb{R}$ are given by
\begin{equation}
V_{m,0} \triangleq \eta \frac{v_{m,1,0}}{v_0} + \eta_c \frac{v_{m,2,0}}{v_0} 
\end{equation}
\begin{equation}
U_1(T) \triangleq \frac{T_c^2}{X_c} \frac{1}{m_1 + m_2} u_1 (T_c T)
\end{equation}
\begin{equation}
U_2(T) \triangleq \frac{T_c^2}{X_c} \frac{1}{m_1 + m_2} u_2 (T_c T)
\end{equation}
respectively. Then,
\begin{equation}\label{eq:X_m_1}
X_{m,1} \triangleq X_m + \eta_c X
\end{equation}
\begin{equation}\label{eq:X_m_2}
X_{m,2} \triangleq X_m - \eta X
\end{equation}
\begin{equation}\label{eq:V_m_1}
V_{m,1} \triangleq V_m + \eta_c V
\end{equation}
\begin{equation}\label{eq:V_m_2}
V_{m,2} \triangleq V_m - \eta V
\end{equation}
It should be remarked that, without knowing $u_1$ and $u_2$, natural scales for variables that appear in Eqs. \eqref{eq:com} - \eqref{eq:v_m_2} can hardly be established. The scales $T_c$ and $X_c$ from the BWSHCCM/BWMCM were adopted for compatibility.

\begin{figure}
\centering
\begin{tikzpicture}[>={Stealth[scale=0.6]}]

	\coordinate (O) at (0, 0);
	\coordinate (G1) at (0, 0.85);
	\coordinate (G2) at (0, -1);
	\coordinate (E1) at (0, 1.25);
	\coordinate (E2) at (0, -2);

    \draw[fill=black](0,0) circle (1 pt) node [above  right] {$C$};
    \draw[fill=black](G1) circle (1 pt) node [left] {$G_1$};
    \draw[fill=black](G2) circle (1 pt) node [left] {$G_2$};
    \draw[fill=black](E1) circle (1 pt) node [left] {$E_1$};
    \draw[fill=black](E2) circle (1 pt) node [left] {$E_2$};

    \draw[->, line width=1](3,-3) -- (3,-2) node [left] {$\hat{\mathbf{n}}_1$};
    \draw[->, line width=1](3,-3) -- (2,-3) node [above] {$\hat{\mathbf{n}}_2$};

    \draw[dotted, line width=1](-2.5,0) -- (2.5,0);
    \draw[dotted, line width=1](0,2) node [right] {$A'$} -- (0,-3.5) node [right] {$B'$};    

    \draw[dotted, line width=1](G1) -- ($(G1) + (2.5, 0)$);
    \draw[dotted, line width=1](G2) -- ($(G2) + (2.5, 0)$);
    \draw[<->, dotted, line width=1]($(G1) + (2, 0)$) -- (2, 0) node [midway, right] {$l_1$};
    \draw[<->, dotted, line width=1]($(G2) + (2, 0)$) -- (2, 0) node [midway, right] {$l_2$};

    \draw[->, line width=1](O) -- (0,0.5) node [left] {$\mathbf{F}$};
    \draw[->, line width=1](O) -- (0,-0.5) node [left] {$\mathbf{F}$};

    \draw[->, line width=1](E1) -- ($(E1) + (0, 0.5)$) node [left] {$\mathbf{u_1}$};
    \draw[->, line width=1](E2) -- ($(E2) + (0, 0.75)$) node [right] {$\mathbf{u_2}$};
    
    \draw[->, line width=0.75, dashed](-2,1.25) node [left] {$\mathcal{B}_1$} -- (-1.25,1);
    \draw[->, line width=0.75, dashed](-1.5,-1.5) node [left] {$\mathcal{B}_2$} -- (-0.5,-1.5);
    
    \draw plot [smooth cycle, tension=.8] coordinates {(0,0) (1.5,1) (0,1.5) (-1.5,1)};
    \draw plot [smooth cycle, tension=.8] coordinates {(0,0) (0.75,-2) (0,-3) (-0.75,-2)};
    
    \draw (0,0.2) -- (-0.2,0.2) -- (-0.2,0);
    \draw (3,-2.8) -- (2.8,-2.8) -- (2.8,-3);

\end{tikzpicture}
\caption{Bodies $\mathcal{B}_1$ and $\mathcal{B}_2$ at the time of the collision}\label{fig:nc}
\end{figure}

\section{The Bouc-Wen-Simon-Hunt-Crossley Collision Model}\label{sec:BWSHCCM}

Taking into account the amendments to the model of the physical system presented in Eq. \eqref{eq:main}, the BWSHCCM is stated as
\begin{equation}\label{eq:EBWSHCCM}
\begin{cases}
\dot{x} = v \\
\dot{z} = A v - \beta \abs{z}^{n - 1} z \abs{v} - \gamma \abs{z}^n v \\
\dot{v} = - \alpha \frac{k}{m} \abs{x}^{p - 1} x - \alpha_c \frac{k}{m} \abs{z}^{p - 1} z - \frac{c}{m} \abs{x}^p v + \frac{1}{m} u \\
\begin{matrix} x(0) = 0, & z(0) = 0, & v(0) = -v_0 \end{matrix}
\end{cases}
\end{equation}
The contact force $F : \mathbb{R}^3 \longrightarrow \mathbb{R}$ is given by
\begin{equation}
F(x, z, v) \triangleq - \alpha k \abs{x}^{p - 1} x - \alpha_c k \abs{z}^{p - 1} z - c \abs{x}^p v
\end{equation}
for all $x, z, v \in \mathbb{R}$. Then, given a solution of the BWSHCCM, the time of the separation and the coefficient of restitution can be found via Eq. \eqref{eq:main_t_s} and Eq. \eqref{eq:main_cor}, respectively.

Introduction of the function $U : \mathbb{R}_{\geq 0} \longrightarrow \mathbb{R}$ given by 
\begin{equation}\label{eq:EBWSHCCM_U_nd}
U(T) \triangleq \left( \frac{1}{\alpha + \alpha_c A^p} \right)^{\frac{1}{p + 1}} (m^p k)^{-\frac{1}{p + 1}} v_0^{-\frac{2 p}{p + 1}} u(T_c T)
\end{equation}
for all $T \in \mathbb{R}_{\geq 0}$, and nondimensionalization of the BWSHCCM using the methodology presented in Ref. \cite{logan_applied_2013} and the parameters listed in Table \ref{tab:ND} results in the new form of the NDBWSHCCM:
\begin{equation}\label{eq:EBWSHCCM_nd}
\begin{cases}
\dot{X} = V \\
\dot{Z} = V - B \abs{Z}^{n - 1} Z \abs{V} - \mathit{\Gamma} \abs{Z}^n V \\
\dot{V} = - \kappa \abs{X}^{p - 1} X - \kappa_c \abs{Z}^{p - 1} Z - \sigma \abs{X}^p V + U \\
\begin{matrix} X(0) = 0, & Z(0) = 0, & V(0) = -1 \end{matrix}
\end{cases}
\end{equation}
It should be remarked that $U \in \mathcal{U}_1$ if $u \in \mathcal{U}_1$ and $U \in \mathcal{U}_{\infty}$ if $u \in \mathcal{U}_{\infty}$.

Under the assumption that $B \in \mathbb{R}_{\geq 0}$, $\mathit{\Gamma} \in [-B, B]$, $\kappa \in (0, 1)$, $\sigma \in \mathbb{R}_{\geq 0}$, $n, p \in \mathbb{R}_{\geq 1}$, and $U \in \mathcal{U}_1$, the NDBWSHCCM has a unique bounded (forward in time) solution that can be extended to infinity (see Appendix \ref{sec:ABWSHCCM}). If $U \in \mathcal{U}_{\infty}$, then any restriction of $U$ to $[0, T]$ with $T \in \mathbb{R}_{\geq 0}$ can be continued to a signal in $\mathcal{U}_1$. Thus, global existence, uniqueness, and boundedness of solutions of the NDBWSHCCM for $U \in \mathcal{U}_1$ imply global existence and uniqueness of solutions of the NDBWSHCCM for any $U \in \mathcal{U}_{\infty}$.

In Ref. \cite{milehins_boucwen_2025}, the authors provide a relationship that describes the dependence of the parameters of the NDBWSHCCM on $v_0$, which can be useful for applications in model parameter identification studies (see Sec. \ref{sec:MPI}). In this study, it will be assumed that the parameters depend not only on $v_0$, but also on $u$. The new relationship can be described by the function $\mathcal{P} : \mathbb{P}^{*} \times \mathcal{U} \times \mathbb{R}_{>0}  \longrightarrow \mathbb{P} \times \mathcal{U}$ that maps $P^{*} = (B_b, \mathit{\Gamma}_b, \kappa, \sigma_b, n, p, U_b, T_b) \in \mathbb{P}^{*}$, $u \in \mathcal{U}$ and $v_0 \in \mathbb{R}_{> 0}$ to 
\[
\left( B_b v_0^{\frac{2 n}{p + 1}}, \mathit{\Gamma}_b v_0^{\frac{2 n}{p + 1}}, \kappa, \sigma_b v_0, n, p, U'(P^{*}, u, v_0, \cdot) \right) \in \mathbb{P} \times \mathcal{U}
\] 
where $\mathbb{P}^{*} \subseteq \mathbb{R}^8$ consist of all $P^{*} = (B_b, \mathit{\Gamma}_b, \kappa, \sigma_b, n, p, U_b, T_b)$ such that $B_b \in \mathbb{R}_{\geq 0}$, $\mathit{\Gamma}_b \in [-B_b, B_b]$, $\kappa \in (0, 1)$, $\sigma_b \in \mathbb{R}_{\geq 0}$, $n, p \in \mathbb{R}_{\geq 1}$, $U_b, T_b \in \mathbb{R}_{>0}$, $\mathbb{P} \subseteq \mathbb{R}^6$ consists of all admissible parameters $P = (B, \mathit{\Gamma}, \kappa, \sigma, n, p)$ of the NDBWSHCCM, and $U' : \mathbb{P}^{*} \times \mathcal{U} \times \mathbb{R}_{>0} \times \mathbb{R}_{\geq 0} \longrightarrow \mathbb{R}$ is defined via
\begin{equation}
U'(P^{*}, u, v_0, T) \triangleq U_b v_0^{-\frac{2 p }{p + 1}} u \left( T_b v_0^{-\frac{p - 1}{p + 1}} T \right)
\end{equation}
for all $P^{*} \in \mathbb{P}^{*}$, $u \in \mathcal{U}$, $v_0 \in \mathbb{R}_{>0}$ and $T \in \mathbb{R}_{\geq 0}$ such that $p = P_6^{*}$, $U_b = P_7^{*}$, and $T_b = P_8^{*}$. The members of $\mathbb{P}^{*}$ will be referred to as the base parameters: they are merely convenient abstractions for the study of the behavior of a given physical system represented by the NDBWSHCCM with respect to the changes in the initial relative velocity and inputs (e.g., see Sec. \ref{sec:MPI}). Thus, $\mathcal{P}$ maps the base parameters and inputs of the BWSHCCM to the admissible parameters and the admissible inputs of the NDBWSHCCM.

The functions that represent the relationship between the parameters and various physical quantities of interest are also updated. Thus, $\Phi : \mathbb{P} \times \mathcal{U} \times \mathbb{R}_{\geq 0} \longrightarrow \mathbb{R}^3$ is defined in a manner such that $\Phi_{P, U} (T)$ represents the value of the state of the NDBWSHCCM parameterized by $P \in \mathbb{P}$ and the input $U \in \mathcal{U}$ at the time $T \in \mathbb{R}_{\geq 0}$; the contact force $F : \mathbb{P} \times \mathbb{R}^3 \longrightarrow \mathbb{R}$ shall be defined as
\begin{equation}
F_P(X, Z, V) \triangleq - \kappa \abs{X}^{p - 1} X - \kappa_c \abs{Z}^{p - 1} Z - \sigma \abs{X}^p V
\end{equation}
for any $(X, Z, V) \in \mathbb{R}^3$ and $P \in \mathbb{P}$ such that $\kappa = P_3$, $\sigma = P_4$, and $p = P_6$; the time of the separation $T_s : \mathbb{P} \times \mathcal{U} \longrightarrow \mathbb{R}_{> 0} \cup \{ +\infty \}$ shall be defined as
\begin{equation}
T_s (P, U) \triangleq \inf \{ T \in \mathbb{R}_{\geq 0} : F_P(\Phi_{P, U}(T)) \leq 0 \leq \Phi_{P, U, 3}(T) \}
\end{equation}
for all $P \in \mathbb{P}$ and $U \in \mathcal{U}$; as explained previously, in the context of this study, the duration of the collision $T_d$ is equivalent to the time of the separation $T_s$; CoR $e : \mathbb{P} \times \mathcal{U} \longrightarrow \mathbb{R}$ shall be defined as
\begin{equation}
e(P, U) \triangleq 
\begin{cases}
\Phi_{P,U,3}(T_s(P, U)) & T_s(P, U) \neq +\infty \\
0 & T_s(P, U) = +\infty
\end{cases}
\end{equation}
for all $P \in \mathbb{P}$ and $U \in \mathcal{U}$.

\section{The Bouc-Wen-Maxwell Collision Model}\label{sec:BWMCM}

Taking into account the amendments to the model of the physical system in Eq. \eqref{eq:main}, the BWMCM can be stated as
\begin{equation}\label{eq:EBWMCM}
\begin{cases}
\dot{r} = \alpha \frac{k}{c} \abs{y}^{p - 1} y + \alpha_c \frac{k}{c} \abs{z}^{p - 1} z\\
\dot{y} = - \alpha \frac{k}{c} \abs{y}^{p - 1} y - \alpha_c \frac{k}{c} \abs{z}^{p - 1} z + v\\
\dot{z} = A \dot{y} - \beta \abs{z}^{n - 1} z \abs{\dot{y}} - \gamma \abs{z}^n \dot{y} \\
\dot{v} = - \alpha \frac{k}{m} \abs{y}^{p - 1} y - \alpha_c \frac{k}{m} \abs{z}^{p - 1} z + \frac{1}{m} u \\
\begin{matrix} r(0) = y(0) = z(0) = 0, & v(0) = -v_0 \end{matrix}
\end{cases}
\end{equation}
Then, the relative displacement of the colliding bodies can be recovered via $x = r + y$.
The contact force $F : \mathbb{R}^4 \longrightarrow \mathbb{R}$ is given by 
\begin{equation}
F(r, y, z, v) \triangleq - \alpha k \abs{y}^{p - 1} y - \alpha_c k \abs{z}^{p - 1} z
\end{equation}
for all $r, y, z, v \in \mathbb{R}$. Then, given a solution of the BWMCM, the time of the separation and the coefficient of restitution can be found via Eq. \eqref{eq:main_t_s} and Eq. \eqref{eq:main_cor}, respectively.

It should be noted that the form of the BWMCM given by Eq. \eqref{eq:EBWMCM} differs from the form of the BWMCM that was employed in Ref. \cite{milehins_boucwen_2025} and given by Eq. \eqref{eq:BWMCM}. However, these two forms are equivalent. Suppose that the output function of the model described by Eq. \eqref{eq:EBWMCM} is given by
\begin{equation}\label{eq:EBWMCM_output}
(r, y, z, v) \mapsto (r, y, z, \dot{y}, r + y, v) \triangleq (r, y, z, w, x, v)
\end{equation}
Suppose also that the effects of external forces are ignored ($u = 0$) in Eq. \eqref{eq:EBWMCM}. Then, the model given by Eq. \eqref{eq:BWMCM} and Eq. \eqref{eq:BWMCM_output}, and the model given by Eq. \eqref{eq:EBWMCM} and Eq. \eqref{eq:EBWMCM_output} yield identical outputs. The primary advantage of the form of the model given by Eq. \eqref{eq:EBWMCM} is that the state function associated with the model is locally Lipschitz continuous in the state variables and continuous in the time variable under all admissible parameterizations with $p \in \mathbb{R}_{\geq 1}$. This guarantees uniqueness (and local existence) of the solutions (and, thence, the output) of both models.\footnote{The choice of the form of the model in Ref. \cite{milehins_boucwen_2025} was suboptimal for the target applications.} Similar results can be established for the nondimensionalized model (see below).

Introduction of the function $U : \mathbb{R}_{\geq 0} \longrightarrow \mathbb{R}$ given by
\begin{equation}\label{eq:EBWMCM_U_nd}
U(T) \triangleq \left( \frac{1}{\alpha + \alpha_c A^p} \right)^{\frac{1}{p + 1}} (m^p k)^{-\frac{1}{p + 1}} v_0^{-\frac{2 p }{p + 1}} u(T_c T)
\end{equation}
for all $T \in \mathbb{R}_{\geq 0}$, and nondimensionalization of the BWMCM using the methodology presented in Ref. \cite{logan_applied_2013} and the parameters listed in Table \ref{tab:ND} results in the new form of the NDBWMCM:
\begin{equation}\label{eq:EBWMCM_nd}
\begin{cases}
\dot{R} = \kappa \sigma \abs{Y}^{p - 1} Y + \kappa_c \sigma \abs{Z}^{p - 1} Z\\
\dot{Y} = - \kappa \sigma \abs{Y}^{p - 1} Y - \kappa_c \sigma  \abs{Z}^{p - 1} Z + V\\
\dot{Z} = \dot{Y} - B \abs{Z}^{n - 1} Z \abs{\dot{Y}} - \mathit{\Gamma} \abs{Z}^n \dot{Y} \\
\dot{V} = - \kappa \abs{Y}^{p - 1} Y - \kappa_c \abs{Z}^{p - 1} Z + U \\
\begin{matrix} R(0) = Y(0) = Z(0) = 0, & V(0) = -1 \end{matrix}
\end{cases}
\end{equation}
Then, the relative displacement of the colliding bodies can be recovered via $X = R + Y$. It should be remarked that $U \in \mathcal{U}_1$ if $u \in \mathcal{U}_1$ and $U \in \mathcal{U}_{\infty}$ if $u \in \mathcal{U}_{\infty}$.

Under the assumption that $B \in \mathbb{R}_{\geq 0}$, $\mathit{\Gamma} \in [-B, B]$, $\kappa \in (0, 1)$, $\sigma \in \mathbb{R}_{> 0}$, $n, p \in \mathbb{R}_{\geq 1}$, and $U \in \mathcal{U}_1$, the NDBWMCM has unique bounded (forward in time) solutions that can be extended to infinity (see Appendix \ref{sec:ABWMCM}). If $U \in \mathcal{U}_{\infty}$, then any restriction of $U$ to $[0, T]$ with $T \in \mathbb{R}_{\geq 0}$ can be continued to a signal in $\mathcal{U}_1$. Thus, global existence, uniqueness, and boundedness of solutions of the NDBWMCM for $U \in \mathcal{U}_1$ imply global existence and uniqueness of solutions of the NDBWMCM for any $U \in \mathcal{U}_{\infty}$.

The new relationship between the parameters of the NDBWMCM, $u$ and $v_0$ (see Sec. \ref{sec:BWSHCCM}) is described by the function $\mathcal{P} : \mathbb{P}^{*} \times \mathcal{U} \times \mathbb{R}_{>0}  \longrightarrow \mathbb{P} \times \mathcal{U}$ that maps $P^{*} = (B_b, \mathit{\Gamma}_b, \kappa, \sigma_b, n, p, U_b, T_b) \in \mathbb{P}^{*}$, $u \in \mathcal{U}$ and $v_0 \in \mathbb{R}_{> 0}$ to 
\begin{equation}
\left( B_b v_0^{\frac{2 n}{p + 1}}, \mathit{\Gamma}_b v_0^{\frac{2 n}{p + 1}}, \kappa, \sigma_b v_0^{\frac{p - 1}{p + 1}}, n, p, U'(P^{*}, u, v_0, \cdot) \right) \in \mathbb{P} \times \mathcal{U}
\end{equation}
where $\mathbb{P}^{*} \subseteq \mathbb{R}^8$ consist of all $P^{*} = (B_b, \mathit{\Gamma}_b, \kappa, \sigma_b, n, p, U_b, T_b)$ such that $B_b \in \mathbb{R}_{\geq 0}$, $\mathit{\Gamma}_b \in [-B_b, B_b]$, $\kappa \in (0, 1)$, $\sigma_b \in \mathbb{R}_{> 0}$, $n, p \in \mathbb{R}_{\geq 1}$, $U_b, T_b \in \mathbb{R}_{>0}$, $\mathbb{P} \subseteq \mathbb{R}^6$ consists of all admissible parameters $P = (B, \mathit{\Gamma}, \kappa, \sigma, n, p)$ of the NDBWMCM, and $U' : \mathbb{P}^{*} \times \mathcal{U} \times \mathbb{R}_{>0} \times \mathbb{R}_{\geq 0} \longrightarrow \mathbb{R}$ is given by 
\begin{equation}
U'(P^{*}, u, v_0, T) \triangleq U_b v_0^{-\frac{2 p }{p + 1}} u \left( T_b v_0^{-\frac{p - 1}{p + 1}} T \right)
\end{equation}
for all $P^{*} \in \mathbb{P}^{*}$, $u \in \mathcal{U}$, $v_0 \in \mathbb{R}_{>0}$ and $T \in \mathbb{R}_{\geq 0}$ such that $p = P_6^{*}$, $U_b = P_7^{*}$, and $T_b = P_8^{*}$. 

Furthermore, $\Phi : \mathbb{P} \times \mathcal{U} \times \mathbb{R}_{\geq 0} \longrightarrow \mathbb{R}^4$ is defined in a manner such that $\Phi_{P, U} (T)$ represents the value of the state of the NDBWMCM parameterized by $P \in \mathbb{P}$ and the input $U \in \mathcal{U}$ at the time $T \in \mathbb{R}_{\geq 0}$; the contact force $F : \mathbb{P} \times \mathbb{R}^4 \longrightarrow \mathbb{R}$ shall be defined as
\begin{equation}
F_P(R, Y, Z, V) \triangleq - \kappa \abs{Y}^{p - 1} Y - \kappa_c \abs{Z}^{p - 1} Z
\end{equation}
for any $(R, Y, Z, V) \in \mathbb{R}^4$ and $P \in \mathbb{P}$ such that $\kappa = P_3$ and $p = P_6$; the time of the separation $T_s : \mathbb{P} \times \mathcal{U} \longrightarrow \mathbb{R}_{> 0} \cup \{ +\infty \}$ shall be defined as
\begin{equation}
T_s (P, U) \triangleq \inf \{ T \in \mathbb{R}_{\geq 0} : F_P(\Phi_{P, U}(T)) \leq 0 \leq \Phi_{P,U,4}(T) \}
\end{equation}
for all $P \in \mathbb{P}$ and $U \in \mathcal{U}$; as explained previously, in the context of this study, the duration of the collision $T_d$ is equivalent to the time of the separation $T_s$; CoR $e : \mathbb{P} \times \mathcal{U} \longrightarrow \mathbb{R}$ shall be defined as
\begin{equation}
e(P, U) \triangleq 
\begin{cases}
\Phi_{P,U,4}(T_s(P, U)) & T_s(P, U) \neq +\infty \\
0 & T_s(P, U) = +\infty
\end{cases}
\end{equation}
for all $P \in \mathbb{P}$ and $U \in \mathcal{U}$.

\section{Application Example}\label{sec:AE}

\begin{figure}
\centering
\begin{tikzpicture}[>={Stealth[scale=0.6]}]

	\coordinate (O) at (0, 0);
	\coordinate (G1) at (0, 1.5);
	\coordinate (G2) at (0, -1.5);

    \draw[fill=black](0,0) circle (1 pt) node [above  right] {$C$};
    \draw[fill=black](G1) circle (1 pt) node [left] {$G_1$};
    \draw[fill=black](G2) circle (1 pt) node [left] {$G_2$};

    \draw[->, line width=1](3,2) -- (3,3) node [left] {$\hat{\mathbf{n}}_1$};
    \draw[->, line width=1](3,2) -- (2,2) node [above] {$\hat{\mathbf{n}}_2$};

    \draw[dotted, line width=1](-2.5,0) -- (2.5,0);
    \draw[dotted, line width=1](0,3.5) node [right] {$A'$} -- (0,-3.5) node [right] {$B'$};    
    
    \draw[->, line width=1](O) -- (0,0.5) node [left] {$\mathbf{F}$};
    \draw[->, line width=1](O) -- (0,-0.5) node [left] {$\mathbf{F}$};
    
    \draw[->, line width=1](G1) -- ($(G1) + (0, 0.5)$) node [right] {$\mathbf{u}_1$};   
    \draw[->, line width=1](G2) -- ($(G2) + (0, 0.75)$) node [right] {$\mathbf{u}_2$};
        
    \draw[->, line width=0.75, dashed](-2,1.25) node [left] {$\mathcal{B}_1$} -- (-1,1.5);
    \draw[->, line width=0.75, dashed](-2,-1.25) node [left] {$\mathcal{B}_2$} -- (-1,-1.5);

    \draw[line width=1] (-1.5,3.5) -- (-1.5,-3.5);
    \draw[line width=1] (1.5,3.5) -- (1.5,-3.5);
    
	\draw (0, 1.5) circle (1.5);
	\draw (0, -1.5) circle (1.5);
    
    \draw (0,0.2) -- (-0.2,0.2) -- (-0.2,0);
    \draw (3,2.2) -- (2.8,2.2) -- (2.8,2);

\end{tikzpicture}
\caption{System diagram for Section \ref{sec:AE}}\label{fig:sd}
\end{figure}

\begin{table}[t]
\caption{Two balls colliding under the action of external forces: parameterization of the BWSHCCM and the BWMCM\label{tab:AE}}
\centering{
\begin{tabular}{c c c}
\toprule
Parameter & BWSHCCM & BWMCM \\
\midrule
$g$      & $9.81 \; (\text{m} \; \text{s}^{-2})$                              & $9.81 \; (\text{m} \; \text{s}^{-2})$\\
$\omega$ & $10^3 \; (\text{s}^{-1})$                                       & $10^3 \; (\text{s}^{-1})$\\
$D$      & $250 \; (\text{m} \; \text{s}^{-2})$                           & $250 \; (\text{m} \; \text{s}^{-2})$\\
$a$      & $0.05 \; (\text{m})$                                            & $0.05 \; (\text{m})$\\
$\rho_1$ & $3 \times 10^3 \; (\text{kg} \; \text{m}^{-3})$                 & $3 \times 10^3 \; (\text{kg} \; \text{m}^{-3})$\\
$\rho_2$ & $10^3 \; (\text{kg} \; \text{m}^{-3})$                          & $10^3 \; (\text{kg} \; \text{m}^{-3})$\\
$k$      & $10^6 \; (\text{kg} \; \text{m}^{-0.5} \; \text{s}^{-2})$     & $10^6 \; (\text{kg} \; \text{m}^{-0.5} \; \text{s}^{-2})$ \\
$c$      & $10^3  \; (\text{kg} \; \text{m}^{-1.5} \; \text{s}^{-1})$      & $10^3 \; (\text{kg} \; \text{s}^{-1})$ \\
$n$      & $1.5 \; (\text{-})$                                             & $1.5 \; (\text{-})$ \\
$p$      & $1.5 \; (\text{-})$                                             & $1.5 \; (\text{-})$ \\
$\alpha$ & $0.5 \; (\text{-})$                                             & $0.5 \; (\text{-})$ \\
$\beta$  & $10^4 \; (\text{m}^{-1.5})$                                     & $10^4 \; (\text{m}^{-1.5})$ \\
$\gamma$ & $-5 \times 10^3 \; (\text{m}^{-1.5})$                           & $-5 \times 10^3 \; (\text{m}^{-1.5})$ \\
$A$      & $1 \; (\text{-})$                                               & $1 \; (\text{-})$ \\
$v_{m,1,0}$  & $-0.5 \; (\text{m} \; \text{s}^{-1})$                       & $-0.5 \; (\text{m} \; \text{s}^{-1})$ \\
$v_{m,2,0}$  & $1.5 \; (\text{m} \; \text{s}^{-1})$                        & $1.5 \; (\text{m} \; \text{s}^{-1})$ \\
\bottomrule
\end{tabular}
}
\end{table}

\begin{figure*}[t!]
\begin{subfigure}[t]{0.475\textwidth}
\centering{
    \begin{tikzpicture}
	\definecolor{clr4}{RGB}{0,0,0}
	\begin{axis}[xlabel={$t \; (\text{s})$}, ylabel={$\text{displacement} \; (\text{m})$}, xticklabel style={/pgf/number format/precision=2, /pgf/number format/fixed}, legend style={at={(0.5,1.01)}, anchor=south, font=\small, legend columns=-1}, line width=1pt, height=1.25in, width=\linewidth]
	\addplot [clr4, mark=none, dotted] table [x=t, y=x_m_1, col sep=comma] {data/TB_BWSHCCM.csv};
	\addlegendentry{$x_{m,1}$};
	\addplot [clr4, mark=none, dashed] table [x=t, y=x_m_2, col sep=comma] {data/TB_BWSHCCM.csv};
	\addlegendentry{$x_{m,2}$};
	\addplot [clr4, mark=none, solid] table [x=t, y=x_m, col sep=comma] {data/TB_BWSHCCM.csv};
	\addlegendentry{$x_m$};
	\addplot [clr4, mark=none, dash dot] table [x=t, y=x, col sep=comma] {data/TB_BWSHCCM.csv};
	\addlegendentry{$x$};			
	\end{axis}
	\end{tikzpicture}
}
\subcaption{BWSHCCM}\label{fig:TB_BWSHCCM_x}
\end{subfigure}
\hfill
\begin{subfigure}[t]{0.475\textwidth} 
\centering{
    \begin{tikzpicture}
	\definecolor{clr4}{RGB}{0,0,0}
	\begin{axis}[xlabel={$t \; (\text{s})$}, ylabel={$\text{displacement} \; (\text{m})$}, xticklabel style={/pgf/number format/precision=2, /pgf/number format/fixed}, legend style={at={(0.5,1.01)}, anchor=south, font=\small, legend columns=-1}, line width=1pt, height=1.25in, width=\linewidth]
	\addplot [clr4, mark=none, dotted] table [x=t, y=x_m_1, col sep=comma] {data/TB_BWMCM.csv};
	\addlegendentry{$x_{m,1}$};
	\addplot [clr4, mark=none, dashed] table [x=t, y=x_m_2, col sep=comma] {data/TB_BWMCM.csv};
	\addlegendentry{$x_{m,2}$}
	\addplot [clr4, mark=none, solid] table [x=t, y=x_m, col sep=comma] {data/TB_BWMCM.csv};
	\addlegendentry{$x_m$};	
	\addplot [clr4, mark=none, dash dot] table [x=t, y=x, col sep=comma] {data/TB_BWMCM.csv};
	\addlegendentry{$x$};		
	\end{axis}
	\end{tikzpicture}
}
\subcaption{BWMCM}\label{fig:TB_BWMCM_x}
\end{subfigure}
\caption{Evolution of the contact interface and the locations of the centers of mass during a collision of two balls under the influence of external forces}\label{fig:TB_x}
\end{figure*}

\begin{figure*}[t!]
\begin{subfigure}[t]{0.475\textwidth}
\centering{
    \begin{tikzpicture}
	\definecolor{clr4}{RGB}{0,0,0}
	\begin{axis}[xlabel={$t \; (\text{s})$}, ylabel={$\text{velocity} \; (\text{m} \; \text{s}^{-1})$}, xticklabel style={/pgf/number format/precision=2, /pgf/number format/fixed}, legend style={at={(0.5,1.01)}, anchor=south, font=\small, legend columns=-1}, line width=1pt, height=1.25in, width=\linewidth]
	\addplot [clr4, mark=none, dotted] table [x=t, y=v_m_1, col sep=comma] {data/TB_BWSHCCM.csv};
	\addlegendentry{$v_{m,1}$};
	\addplot [clr4, mark=none, dashed] table [x=t, y=v_m_2, col sep=comma] {data/TB_BWSHCCM.csv};
	\addlegendentry{$v_{m,2}$};
	\addplot [clr4, mark=none, solid] table [x=t, y=v_m, col sep=comma] {data/TB_BWSHCCM.csv};
	\addlegendentry{$v_m$};
	\addplot [clr4, mark=none, dash dot] table [x=t, y=v, col sep=comma] {data/TB_BWSHCCM.csv};
	\addlegendentry{$v$};			
	\end{axis}
	\end{tikzpicture}
}
\subcaption{BWSHCCM}\label{fig:TB_BWSHCCM_v}
\end{subfigure}
\hfill
\begin{subfigure}[t]{0.475\textwidth} 
\centering{
    \begin{tikzpicture}
	\definecolor{clr4}{RGB}{0,0,0}
	\begin{axis}[xlabel={$t \; (\text{s})$}, ylabel={$\text{velocity} \; (\text{m} \; \text{s}^{-1})$}, xticklabel style={/pgf/number format/precision=2, /pgf/number format/fixed}, legend style={at={(0.5,1.01)}, anchor=south, font=\small, legend columns=-1}, line width=1pt, height=1.25in, width=\linewidth]
	\addplot [clr4, mark=none, dotted] table [x=t, y=v_m_1, col sep=comma] {data/TB_BWMCM.csv};
	\addlegendentry{$v_{m,1}$};
	\addplot [clr4, mark=none, dashed] table [x=t, y=v_m_2, col sep=comma] {data/TB_BWMCM.csv};
	\addlegendentry{$v_{m,2}$}
	\addplot [clr4, mark=none, solid] table [x=t, y=v_m, col sep=comma] {data/TB_BWMCM.csv};
	\addlegendentry{$v_m$};	
	\addplot [clr4, mark=none, dash dot] table [x=t, y=v, col sep=comma] {data/TB_BWMCM.csv};
	\addlegendentry{$v$};		
	\end{axis}
	\end{tikzpicture}
}
\subcaption{BWMCM}\label{fig:TB_BWMCM_v}
\end{subfigure}
\caption{Evolution of the relative velocity at the contact interface and the velocities of the centers of mass during a collision of two balls under the influence of external forces}\label{fig:TB_v}
\end{figure*}

\begin{figure*}[t!]
\begin{subfigure}[t]{0.475\textwidth}
\centering{
    \begin{tikzpicture}
	\definecolor{clr4}{RGB}{0,0,0}
	\begin{axis}[xlabel={$t \; (\text{s})$}, ylabel={$\text{force} \; (\text{N})$}, xticklabel style={/pgf/number format/precision=2, /pgf/number format/fixed}, legend style={at={(0.5,1.01)}, anchor=south, font=\small, legend columns=-1}, line width=1pt, height=1.25in, width=\linewidth]
	\addplot [clr4, mark=none, dotted] table [x=t, y=u_1, col sep=comma] {data/TB_BWSHCCM.csv};
	\addlegendentry{$u_1$};
	\addplot [clr4, mark=none, dashed] table [x=t, y=u_2, col sep=comma] {data/TB_BWSHCCM.csv};
	\addlegendentry{$u_2$};
	\addplot [clr4, mark=none, solid] table [x=t, y=u, col sep=comma] {data/TB_BWSHCCM.csv};
	\addlegendentry{$u$};
	\addplot [clr4, mark=none, dash dot] table [x=t, y=F, col sep=comma] {data/TB_BWSHCCM.csv};
	\addlegendentry{$F$};			
	\end{axis}
	\end{tikzpicture}
}
\subcaption{BWSHCCM}\label{fig:TB_BWSHCCM_F}
\end{subfigure}
\hfill
\begin{subfigure}[t]{0.475\textwidth} 
\centering{
    \begin{tikzpicture}
	\definecolor{clr4}{RGB}{0,0,0}
	\begin{axis}[xlabel={$t \; (\text{s})$}, ylabel={$\text{force} \; (\text{N})$}, xticklabel style={/pgf/number format/precision=2, /pgf/number format/fixed}, legend style={at={(0.5,1.01)}, anchor=south, font=\small, legend columns=-1}, line width=1pt, height=1.25in, width=\linewidth]
	\addplot [clr4, mark=none, dotted] table [x=t, y=u_1, col sep=comma] {data/TB_BWMCM.csv};
	\addlegendentry{$u_1$};
	\addplot [clr4, mark=none, dashed] table [x=t, y=u_2, col sep=comma] {data/TB_BWMCM.csv};
	\addlegendentry{$u_2$}
	\addplot [clr4, mark=none, solid] table [x=t, y=u, col sep=comma] {data/TB_BWMCM.csv};
	\addlegendentry{$u$};	
	\addplot [clr4, mark=none, dash dot] table [x=t, y=F, col sep=comma] {data/TB_BWMCM.csv};
	\addlegendentry{$F$};		
	\end{axis}
	\end{tikzpicture}
}
\subcaption{BWMCM}\label{fig:TB_BWMCM_F}
\end{subfigure}
\caption{Evolution of the contact force and the external forces during a collision of two balls under the influence of external forces}\label{fig:TB_F}
\end{figure*}

The discussion that follows is with reference to Fig. \ref{fig:sd}. Assume that two homogeneous balls, indexed 1 and 2, constrained to move in a vertical pipe near the surface of the Earth without friction or rotation collide at the time $t_0 = 0 \in \mathbb{R}_{\geq 0}$ at the point $C = (0, 0, 0)$. The radii of the balls are, effectively, equal to the radius of the tube, $a \in \mathbb{R}_{>0}$. The mass of ball $i$ is given by $m_i \triangleq (4/3) \pi \rho_i a^3$ (e.g., see Ref. \cite{shurman_calculus_2016} and Ref. \cite{morro_mathematical_2023}) where $\rho_i \in \mathbb{R}_{>0}$ is the density of the material of ball $i$. Ball 1 is acted upon by the force of gravity given by
\begin{equation}
\mathbf{u}_1(t) \triangleq u_1(t) \hat{\mathbf{n}}_1 \triangleq -m_1 g \hat{\mathbf{n}}_1
\end{equation}
for all $t \in \mathbb{R}_{\geq 0}$. The point of application of the resultant of all external forces acting on ball 2 is its center of mass $G_2$. It is given by 
\begin{equation}
\mathbf{u}_2(t) \triangleq u_2(t) \hat{\mathbf{n}}_1 \triangleq \frac{2}{\pi} m_2 D \arctan{(\omega t)} \hat{\mathbf{n}}_1
\end{equation}

All relevant parameters for the models are listed in Table \ref{tab:AE}. Figures \ref{fig:TB_x}, \ref{fig:TB_v}, and \ref{fig:TB_F} show the evolution of relevant positions, velocities, and forces acting on the bodies during the collision, respectively, obtained based on the (rescaled) results of the numerical simulations of the NDBWSHCCM and the NDBWMCM. The methodology that was used for the simulation and data analysis is outlined in Appendix \ref{sec:SDA}.

\section{Model Parameter Identification}\label{sec:MPI}

\begin{table*}[t]
\caption{DSS and DSA: model parameter identification: the columns labeled DSS (0) and DSA (0) provide the data that were obtained based on the results of the previous study, the columns labeled DSS (1) and DSA (1) provide the data that were obtained based on the results of the present study; the values of $U_b$ and $T_b$ are irrelevant as it was assumed that $u = 0$; the dimensional parameters are stated in the SI base units \label{tab:DSSDSA_param}}
\centering{
\begin{tabular}{c c c c c c c c c}
\toprule
& \multicolumn{4}{c}{NDBWSHCCM} & \multicolumn{4}{c}{NDBWMCM}\\
\midrule
$P^{*}$ \& $J$      & DSS (0)   & DSS (1)    & DSA (0)    & DSA (1)  & DSS (0)  & DSS (1)                 & DSA (0)    & DSA (1) \\
\midrule
$B_b$               & $1.43$    & $2.38$     & $0.63$     & $1.04$   & $0.655$  & $0.966$                 & $0.44$     & $0.521$ \\
$\mathit{\Gamma}_b$ & $-1.42$   & $-2.38$    & $-0.611$   & $-1.02$  & $-0.64$  & $-0.966$                & $-0.418$   & $-0.521$ \\
$\kappa$            & $0.632$   & $0.677$    & $0.188$    & $0.362$  & $0.519$  & $0.531$                 & $0.113$    & $0.168$ \\
$\sigma_b$          & $0.00715$ & $0.0348$   & $0.00594$  & $0.017$  & $0.0118$ & $2.22 \times 10^{-16}$  & $0.00785$  & $0.0144$ \\
$n$                 & $1.31$    & $2.93$     & $1$        & $1.01$   & $1.94$   & $1.57$                  & $1.27$     & $1$ \\
$p$                 & $1.27$    & $2.62$     & $2.02$     & $1.67$   & $2.28$   & $1.75$                  & $3.14$     & $2.04$ \\
$U_b$               & n/a       & n/a        & n/a        & n/a      & n/a      & n/a                     & n/a        & n/a \\
$T_b$               & n/a       & n/a        & n/a        & n/a      & n/a      & n/a                     & n/a        & n/a \\
\midrule
$J_1 (\cdot)$ & $5.81 \times 10^{-5}$ & $3.51 \times 10^{-5}$ & $7.62 \times 10^{-6}$ & $7.23 \times 10^{-6}$ & $7.49 \times 10^{-5}$ & $6.05 \times 10^{-5}$ & $8.07 \times 10^{-6}$ & $8.02 \times 10^{-6}$ \\
\bottomrule
\end{tabular}
}
\end{table*}

\begin{figure*}[t]
\begin{subfigure}[t]{0.475\textwidth}
\centering{
    \begin{tikzpicture}
	\definecolor{clr1}{RGB}{150,150,150}
	\definecolor{clr2}{RGB}{0,0,0}
	\begin{axis}[xlabel={$v_0 \; (\text{m} \; \text{s}^{-1})$}, ylabel={$e$}, ymin = 0.5, ymax = 1, xticklabel style={/pgf/number format/precision=2, /pgf/number format/fixed}, legend style={font=\small}, line width=1pt]
	\addplot [clr1, mark=o, mark options={scale=1.5}, only marks] table [x=v_0, y=e, col sep=comma] {data/DSS.csv};
	\addlegendentry{DSS: experiment};
	\addplot [clr1, mark=x, mark options={scale=1.5}, only marks] table [x=v_0, y=e, col sep=comma] {data/DSS_NDBWSHCCM_0.csv};
	\addlegendentry{DSS: NDBWSHCCM (0)};	
	\addplot [clr1, mark=+, mark options={scale=1.5}, only marks] table [x=v_0, y=e, col sep=comma] {data/DSS_NDBWSHCCM_1.csv};
	\addlegendentry{DSS: NDBWSHCCM (1)};		
	\addplot [clr2, mark=o, mark options={scale=1.5}, only marks] table [x=v_0, y=e, col sep=comma] {data/DSA.csv};
	\addlegendentry{DSA: experiment};
	\addplot [clr2, mark=x, mark options={scale=1.5}, only marks] table [x=v_0, y=e, col sep=comma] {data/DSA_NDBWSHCCM_0.csv};
	\addlegendentry{DSA: NDBWSHCCM (0)};
	\addplot [clr2, mark=+, mark options={scale=1.5}, only marks] table [x=v_0, y=e, col sep=comma] {data/DSA_NDBWSHCCM_1.csv};
	\addlegendentry{DSA: NDBWSHCCM (1)};
	\end{axis}
	\end{tikzpicture}
}
\subcaption{CoR: NDBWSHCCM vs. experiment: NDBWSHCCM (0) represents the data from the previous study, NDBWSHCCM (1) represents the data from the present study}\label{fig:CoR_NDBWSHCCM}
\end{subfigure}
\hfill
\begin{subfigure}[t]{0.475\textwidth} 
\centering{
    \begin{tikzpicture}
	\definecolor{clr1}{RGB}{150,150,150}
	\definecolor{clr2}{RGB}{0,0,0}    
	\begin{axis}[xlabel={$v_0 \; (\text{m} \; \text{s}^{-1})$}, ylabel={$e$}, ymin = 0.5, ymax = 1, xticklabel style={/pgf/number format/precision=2, /pgf/number format/fixed}, legend style={font=\small}, line width=1pt]
	\addplot [clr1, mark=o, mark options={scale=1.5}, only marks] table [x=v_0, y=e, col sep=comma] {data/DSS.csv};
	\addlegendentry{DSS: experiment};
	\addplot [clr1, mark=x, mark options={scale=1.5}, only marks] table [x=v_0, y=e, col sep=comma] {data/DSS_NDBWMCM_0.csv};
	\addlegendentry{DSS: NDBWMCM (0)};
	\addplot [clr1, mark=+, mark options={scale=1.5}, only marks] table [x=v_0, y=e, col sep=comma] {data/DSS_NDBWMCM_1.csv};
	\addlegendentry{DSS: NDBWMCM (1)};	
	\addplot [clr2, mark=o, mark options={scale=1.5}, only marks] table [x=v_0, y=e, col sep=comma] {data/DSA.csv};
	\addlegendentry{DSA: experiment};
	\addplot [clr2, mark=x, mark options={scale=1.5}, only marks] table [x=v_0, y=e, col sep=comma] {data/DSA_NDBWMCM_0.csv};
	\addlegendentry{DSA: NDBWMCM (0)};
	\addplot [clr2, mark=+, mark options={scale=1.5}, only marks] table [x=v_0, y=e, col sep=comma] {data/DSA_NDBWMCM_1.csv};
	\addlegendentry{DSA: NDBWMCM (1)};
	\end{axis}
	\end{tikzpicture}
}
\subcaption{CoR: NDBWMCM vs. experiment: NDBWMCM (0) represents the data from the previous study, NDBWMCM (1) represents the data from the present study}\label{fig:CoR_NDBWMCM}
\end{subfigure}
\caption{Kharaz and Gorham (2000): CoR: models vs. experiment}\label{fig:CoR}
\end{figure*}

\begin{table*}[t]
\caption{Normal impact of a baseball on a flat surface: parameterization of the BWSHCCM and the BWMCM\label{tab:RCross}}
\centering{
\begin{tabular}{c c c}
\toprule
Parameter & BWSHCCM & BWMCM \\
\midrule
$m$      & $0.145 \; (\text{kg})$                                          & $0.145 \; (\text{kg})$ \\
$k$      & $117080063 \; (\text{kg} \; \text{m}^{1 - p} \; \text{s}^{-2})$ & $253000000 \; (\text{kg} \; \text{m}^{1 - p} \; \text{s}^{-2})$ \\
$c$      & $5854003  \; (\text{kg} \; \text{m}^{-p} \; \text{s}^{-1})$     & $2811 \; (\text{kg} \; \text{s}^{-1})$ \\
$n$      & $1.1 \; (\text{-})$                                             & $1.2 \; (\text{-})$ \\
$p$      & $1.7 \; (\text{-})$                                             & $1.8 \; (\text{-})$ \\
$\alpha$ & $0.1 \; (\text{-})$                                             & $0.15 \; (\text{-})$ \\
$\beta$  & $981.05 \; (\text{m}^{-n})$                                     & $1200 \; (\text{m}^{-n})$ \\
$\gamma$ & $-961.4 \; (\text{m}^{-n})$                                     & $-1200 \; (\text{m}^{-n})$ \\
$A$      & $0.925 \; (\text{-})$                                           & $1.01 \; (\text{-})$ \\
\bottomrule
\end{tabular}
}
\end{table*}

\begin{figure*}[t]
\begin{subfigure}[t]{0.475\textwidth}
\centering{
    \begin{tikzpicture}
	\definecolor{clr1}{RGB}{200,200,200}
	\definecolor{clr2}{RGB}{150,150,150}
	\definecolor{clr3}{RGB}{100,100,100}
	\definecolor{clr4}{RGB}{0,0,0}
	\begin{axis}[xlabel={$\abs{x} \; (\text{mm})$}, ylabel={$F \; (\text{N})$}, ymin = 0, xticklabel style={/pgf/number format/precision=2, /pgf/number format/fixed}, legend pos=north west, legend style={font=\small}, line width=1pt]
	\addplot [clr1, mark=none, dashed] table [x index=0, y index=1, col sep=comma] {data/RCross_2_15.csv};
	\addlegendentry{$v_0 = 2.15 \; (\text{m} \; \text{s}^{-1})$};
	\addplot [clr2, mark=none, dashed] table [x index=0, y index=1, col sep=comma] {data/RCross_3_03.csv};
	\addlegendentry{$v_0 = 3.03 \; (\text{m} \; \text{s}^{-1})$};
	\addplot [clr3, mark=none, dashed] table [x index=0, y index=1, col sep=comma] {data/RCross_4_18.csv};
	\addlegendentry{$v_0 = 4.18 \; (\text{m} \; \text{s}^{-1})$};
	\addplot [clr4, mark=none, dashed] table [x index=0, y index=1, col sep=comma] {data/RCross_5_02.csv};
	\addlegendentry{$v_0 = 5.02 \; (\text{m} \; \text{s}^{-1})$};
	\addplot [clr1, mark=none, mark options={scale=0.75}] table [x expr=1000*\thisrowno{0}, y index=1, col sep=comma] {data/RCross_BWSHCCM_sim_2_15.csv};
	\addlegendentry{$v_0 = 2.15 \; (\text{m} \; \text{s}^{-1})$};
	\addplot [clr2, mark=none, mark options={scale=0.75}] table [x expr=1000*\thisrowno{0}, y index=1, col sep=comma] {data/RCross_BWSHCCM_sim_3_03.csv};
	\addlegendentry{$v_0 = 3.03 \; (\text{m} \; \text{s}^{-1})$};
	\addplot [clr3, mark=none, mark options={scale=0.75}] table [x expr=1000*\thisrowno{0}, y index=1, col sep=comma] {data/RCross_BWSHCCM_sim_4_18.csv};
	\addlegendentry{$v_0 = 4.18 \; (\text{m} \; \text{s}^{-1})$};
	\addplot [clr4, mark=none, mark options={scale=0.75}] table [x expr=1000*\thisrowno{0}, y index=1, col sep=comma] {data/RCross_BWSHCCM_sim_5_02.csv};
	\addlegendentry{$v_0 = 5.02 \; (\text{m} \; \text{s}^{-1})$};
	\end{axis}
	\end{tikzpicture}
}
\subcaption{Normal impact of a baseball on a flat surface: experimentally obtained hysteresis loops (dashed lines) vs. hysteresis loops obtained from the numerical simulations of the BWSHCCM (solid lines)}\label{fig:RCross_NDBWSHCCM}
\end{subfigure}
\hfill
\begin{subfigure}[t]{0.475\textwidth} 
\centering{
    \begin{tikzpicture}
	\definecolor{clr1}{RGB}{200,200,200}
	\definecolor{clr2}{RGB}{150,150,150}
	\definecolor{clr3}{RGB}{100,100,100}
	\definecolor{clr4}{RGB}{0,0,0}
	\begin{axis}[xlabel={$\abs{x} \; (\text{mm})$}, ylabel={$F \; (\text{N})$}, ymin = 0, xticklabel style={/pgf/number format/precision=2, /pgf/number format/fixed}, legend pos=north west, legend style={font=\small}, line width=1pt]
	\addplot [clr1, mark=none, dashed] table [x index=0, y index=1, col sep=comma] {data/RCross_2_15.csv};
	\addlegendentry{$v_0 = 2.15 \; (\text{m} \; \text{s}^{-1})$};
	\addplot [clr2, mark=none, dashed] table [x index=0, y index=1, col sep=comma] {data/RCross_3_03.csv};
	\addlegendentry{$v_0 = 3.03 \; (\text{m} \; \text{s}^{-1})$};
	\addplot [clr3, mark=none, dashed] table [x index=0, y index=1, col sep=comma] {data/RCross_4_18.csv};
	\addlegendentry{$v_0 = 4.18 \; (\text{m} \; \text{s}^{-1})$};
	\addplot [clr4, mark=none, dashed] table [x index=0, y index=1, col sep=comma] {data/RCross_5_02.csv};
	\addlegendentry{$v_0 = 5.02 \; (\text{m} \; \text{s}^{-1})$};
	\addplot [clr1, mark=none, mark options={scale=0.75}] table [x expr=1000*\thisrowno{0}, y index=1, col sep=comma] {data/RCross_BWMCM_sim_2_15.csv};
	\addlegendentry{$v_0 = 2.15 \; (\text{m} \; \text{s}^{-1})$};
	\addplot [clr2, mark=none, mark options={scale=0.75}] table [x expr=1000*\thisrowno{0}, y index=1, col sep=comma] {data/RCross_BWMCM_sim_3_03.csv};
	\addlegendentry{$v_0 = 3.03 \; (\text{m} \; \text{s}^{-1})$};
	\addplot [clr3, mark=none, mark options={scale=0.75}] table [x expr=1000*\thisrowno{0}, y index=1, col sep=comma] {data/RCross_BWMCM_sim_4_18.csv};
	\addlegendentry{$v_0 = 4.18 \; (\text{m} \; \text{s}^{-1})$};
	\addplot [clr4, mark=none, mark options={scale=0.75}] table [x expr=1000*\thisrowno{0}, y index=1, col sep=comma] {data/RCross_BWMCM_sim_5_02.csv};
	\addlegendentry{$v_0 = 5.02 \; (\text{m} \; \text{s}^{-1})$};
	\end{axis}
	\end{tikzpicture}
}
\subcaption{Normal impact of a baseball on a flat surface: experimentally obtained hysteresis loops (dashed lines) vs. hysteresis loops obtained from the numerical simulations of the BWMCM (solid lines)}\label{fig:RCross_NDBWMCM}
\end{subfigure}
\caption{Normal impact of a baseball on a flat surface: models vs. experiment}\label{fig:RCross}
\end{figure*}

\subsection{Background}\label{sec:MPI:Background}

In this section, the two model parameter identification studies that were presented in Ref. \cite{milehins_boucwen_2025} are updated, and an additional model parameter identification study that is based on an experiment that showcases the impact of the effect of external forces on the behavior of bodies during the collision process is presented. 

The experimental data are often provided in the form of a finite sequence of measured absolute values of the initial relative velocities $\tilde{v}_0 \in \mathbb{R}_{>0}^M$ with $M \in \mathbb{Z}_{\geq 1}$ and a finite sequence of corresponding finite sequences of measured aggregate quantities (such as CoR $\tilde{e}$ or the duration of the collision $\tilde{t}_d$) $\tilde{\Theta} \in (\mathbb{R}^M)^N$ with $N \in \mathbb{Z}_{\geq 1}$.\footnote{It should be remarked that, often, only the averaged quantities obtained over multiple realizations are reported upon.} The experimental data may also be provided in the form of a finite sequence of measured absolute values of the initial relative velocities $\tilde{v}_0 \in \mathbb{R}_{>0}^M$ with $M \in \mathbb{Z}_{\geq 1}$ and a finite sequence of corresponding hysteresis loops: $M$ sequences indexed by $j \in \{1, \ldots, M \}$ that contain the contact force data $\tilde{F}_j \in \mathbb{R}^{K_j}$ vs. displacement data $\tilde{x}_j \in \mathbb{R}^{K_j}$ in the chronological order and with each $K_j \in \mathbb{Z}_{\geq 1}$. It will be assumed that the external force $u_j \in \mathcal{U}$ is known exactly for every $j \in \{1, \ldots, M\}$.

It is more convenient to perform the model parameter identification using the nondimensionalized collision models (i.e., the NDBWSHCCM and the NDBWMCM rather than the BWSHCCM and the BWMCM) as they have fewer parameters. Usually, the base parameters will be identified. The physical parameters associated with the BWSHCCM and the BWMCM can be recovered using the relationships presented in Table \ref{tab:ND}, Eq. \eqref{eq:EBWSHCCM_U_nd} and Eq. \eqref{eq:EBWMCM_U_nd}. However, the physical parameters may not be unique for a given vector of base parameters.

It will be assumed that every aggregate quantity of interest, indexed by $j \in \{1, \ldots, N\}$, can be expressed as a function $\Theta_j :  \mathbb{P}^{*} \times \mathcal{U} \times \mathbb{R}_{>0} \longrightarrow \mathbb{R}$. Then, the quality of a base parameterization $P^{*} \in \mathbb{P}^{*}$ of the NDBWSHCCM or the NDBWMCM may be assessed by the cost functions $J_j : \mathbb{R}_{>0}^M \times \mathbb{R}^M \times \mathcal{U}^M \times \mathbb{P}^{*} \longrightarrow \mathbb{R}_{\geq 0}$, one for each $j \in \{1, \ldots, N\}$, given by
\begin{equation}\label{eq:J}
J_j (\tilde{v}_0, \tilde{\Theta}_j, u, P^{*}) \triangleq \frac{1}{M} \sum_{l = 1}^{l = M} (\tilde{\Theta}_{j, l} - \Theta_j (P^{*}, u_l, \tilde{v}_{0, l}))^2
\end{equation}
which provide the mean squared modeling errors. Then, the model parameter identification problem can be stated as a global multi-objective nonlinear constrained optimization problem \cite{edgeworth_mathematical_1881, pareto_manual_2014}
\begin{equation}
{\arg\min}_{P^{*}} \; J_j (\tilde{v}_0, \tilde{\Theta}_j, u, P^{*}), \; j \in \{ 1, \ldots, N \}
\end{equation}
subject to
\begin{equation}\label{eq:op}
\begin{matrix}
0 \leq B_b \leq B_b^{u} \\
- B_b - \mathit{\Gamma}_b \leq 0 \\
- B_b + \mathit{\Gamma}_b \leq 0 \\
\kappa^{l} \leq \kappa \leq \kappa^{u}\\
\sigma_b^{l} \leq \sigma_b \leq \sigma_b^{u}\\
1 \leq p \leq p^{u}\\
1 \leq n \leq n^{u}\\
U_b^{l} \leq U_b \leq U_b^{u}\\
T_b^{l} \leq T_b \leq T_b^{u}
\end{matrix}
\end{equation}
where $P^{*} = (B_b, \mathit{\Gamma}_b, \kappa, \sigma_b, n, p, U_b, T_b)$, and where $\kappa^{l} \in (0, 1)$, $\sigma_b^{l} \in \mathbb{R}_{>0}$,\footnote{$\sigma_b^{l}$ may take the value of $0$ in the context of the identification of the parameters of the NDBWSHCCM, but this distinction has little practical significance.} $U_b^{l} \in \mathbb{R}_{> 0}$, $T_b^{l} \in \mathbb{R}_{> 0}$ are the lower bounds of the parameters, and $B_b^{u} \in \mathbb{R}_{\geq 0}$, $\kappa^{u} \in [\kappa^{l}, 1)$, $\sigma_b^{u} \in \mathbb{R}_{\geq \sigma_b^{l}}$, $n^{u} \in \mathbb{R}_{\geq 1}$, $p^{u} \in \mathbb{R}_{\geq 1}$, $U_b^{u} \in \mathbb{R}_{\geq U_b^{l}}$, $T_b^{u} \in \mathbb{R}_{\geq T_b^{l}}$ are the upper bounds of the parameters. 

Only two aggregate quantities will be considered in this study: CoR $\Theta_1$ given by 
\begin{equation}\label{eq:Theta_1}
\Theta_1(P^{*}, u, v_0) \triangleq e(\mathcal{P}(P^{*}, u, v_0)))
\end{equation}
and the duration of the collision $\Theta_2$ given by 
\begin{equation}\label{eq:Theta_2}
\Theta_2(P^{*}, u, v_0) \triangleq T_d (\mathcal{P}(P^{*}, u, v_0))) T_b v_0^{-(p - 1)/(p + 1)}
\end{equation}
with $p = P^{*}_6$ and $T_b = P^{*}_8$. 

The identification of the model parameters based on the hysteresis data was performed using a less formal procedure, and its detailed description will be omitted for brevity.

It should be remarked that due to the nature of the methodology that was chosen for the identification of the models, any apparent discrepancies in the quality of the parameterizations obtained using different models are not indicative of the capabilities of the models at large. The primary goal of the identification studies was to showcase that both models can provide an adequate description of the collision phenomena. 

The methodology that was used for the numerical simulation and data analysis is outlined in Appendix \ref{sec:SDA}.

\begin{figure*}[t]
\begin{subfigure}[t]{0.5\textwidth}
\centering{
    \begin{tikzpicture}
	\definecolor{clrk}{RGB}{0,0,0}
	\begin{axis}[xlabel={$v_0 \; (\text{m} \; \text{s}^{-1})$}, ylabel={$e$}, ymin = 0, ymax = 1, xticklabel style={/pgf/number format/precision=2, /pgf/number format/fixed}, legend style={font=\small}, legend pos=south east, line width=1pt, height=1.5in, width=\linewidth]
	\addplot [clrk, mark=o, mark options={scale=1.5}, only marks] table [x=v_0, y=e, col sep=comma] {data/VRBR.csv};
	\addlegendentry{experiment};
	\addplot [clrk, mark=x, mark options={scale=1.5}, only marks] table [x=v_0, y=e, col sep=comma] {data/VRBR_NDBWSHCCM_0.csv};
	\addlegendentry{NDBWSHCCM};	
	\end{axis}
	\end{tikzpicture}
}
\subcaption{CoR: NDBWSHCCM vs. experiment}\label{fig:VRBR_CoR_NDBWSHCCM}
\end{subfigure}
\begin{subfigure}[t]{0.5\textwidth} 
\centering{
    \begin{tikzpicture}
	\definecolor{clrk}{RGB}{0,0,0}    
	\begin{axis}[xlabel={$v_0 \; (\text{m} \; \text{s}^{-1})$}, ylabel={$e$}, ymin = 0, ymax = 1, xticklabel style={/pgf/number format/precision=2, /pgf/number format/fixed}, legend style={font=\small}, legend pos=south east, line width=1pt, height=1.5in, width=\linewidth]
	\addplot [clrk, mark=o, mark options={scale=1.5}, only marks] table [x=v_0, y=e, col sep=comma] {data/VRBR.csv};
	\addlegendentry{experiment};
	\addplot [clrk, mark=x, mark options={scale=1.5}, only marks] table [x=v_0, y=e, col sep=comma] {data/VRBR_NDBWMCM_0.csv};
	\addlegendentry{NDBWMCM};
	\end{axis}
	\end{tikzpicture}
}
\subcaption{CoR: NDBWMCM vs. experiment}\label{fig:VRBR_CoR_NDBWMCM}
\end{subfigure}
\caption{Villegas et al (2021): CoR: models vs. experiment}\label{fig:VRBR_CoR}
\end{figure*}

\begin{figure*}[t]
\begin{subfigure}[t]{0.5\textwidth}
\centering{
    \begin{tikzpicture}
	\definecolor{clrk}{RGB}{0,0,0}
	\begin{axis}[xlabel={$v_0 \; (\text{m} \; \text{s}^{-1})$}, ylabel={$t_d \; (\text{s})$}, ymin = 5/100, ymax = 1, ymode=log, xticklabel style={/pgf/number format/precision=2, /pgf/number format/fixed}, legend style={font=\small}, legend pos=north east, line width=1pt, height=1.5in, width=\linewidth]
	\addplot [clrk, mark=o, mark options={scale=1.5}, only marks] table [x=v_0, y=t_d, col sep=comma] {data/VRBR.csv};
	\addlegendentry{experiment};
	\addplot [clrk, mark=x, mark options={scale=1.5}, only marks] table [x=v_0, y=t_d, col sep=comma] {data/VRBR_NDBWSHCCM_0.csv};
	\addlegendentry{NDBWSHCCM};	
	\end{axis}
	\end{tikzpicture}
}
\subcaption{$t_d$: NDBWSHCCM vs. experiment}\label{fig:VRBR_t_d_NDBWSHCCM}
\end{subfigure}
\begin{subfigure}[t]{0.5\textwidth} 
\centering{
    \begin{tikzpicture}
	\definecolor{clrk}{RGB}{0,0,0}    
	\begin{axis}[xlabel={$v_0 \; (\text{m} \; \text{s}^{-1})$}, ylabel={$t_d \; (\text{s})$}, ymin = 5/100, ymax = 1, ymode=log, xticklabel style={/pgf/number format/precision=2, /pgf/number format/fixed}, legend style={font=\small}, legend pos=north east, line width=1pt, height=1.5in, width=\linewidth]
	\addplot [clrk, mark=o, mark options={scale=1.5}, only marks] table [x=v_0, y=t_d, col sep=comma] {data/VRBR.csv};
	\addlegendentry{experiment};
	\addplot [clrk, mark=x, mark options={scale=1.5}, only marks] table [x=v_0, y=t_d, col sep=comma] {data/VRBR_NDBWMCM_0.csv};
	\addlegendentry{NDBWMCM};
	\end{axis}
	\end{tikzpicture}
}
\subcaption{$t_d$: NDBWMCM vs. experiment}\label{fig:VRBR_t_d_NDBWMCM}
\end{subfigure}
\caption{Villegas et al (2021): $t_d$: models vs. experiment}\label{fig:VRBR_t_d}
\end{figure*}

\begin{table}[t]
\caption{Villegas et al (2021): model parameter identification (the dimensional parameters are stated in the SI base units) \label{tab:VRBR_param}}
\centering{
\begin{tabular}{c c c}
\toprule
$P^{*}$ \& $J$ & \multicolumn{1}{c}{NDBWSHCCM} & \multicolumn{1}{c}{NDBWMCM}\\
\midrule
$B_b$                & $17.4$     & $14.9$ \\
$\mathit{\Gamma}_b$  & $13.1$     & $14.7$ \\
$\kappa$             & $0.324$    & $0.454$ \\
$\sigma_b$           & $0.0794$   & $0.099$ \\
$n$                  & $1$        & $1.09$ \\
$p$                  & $1$        & $1$ \\
$U_b$                & $0.0503$   & $0.0595$ \\
$T_b$                & $0.0255$   & $0.0301$ \\
\midrule
$J_1 (\cdot)$ & $9.02 \times 10^{-5}$ & $3.58 \times 10^{-4}$ \\
$J_2 (\cdot)$ & $1.02 \times 10^{-4}$ & $1.31 \times 10^{-4}$ \\
\bottomrule
\end{tabular}
}
\end{table}

\subsection{Kharaz and Gorham (2000)}\label{sec:MPI:KG}

The first parameter identification study that was presented in Ref. \cite{milehins_boucwen_2025} employed the experimental datasets provided in Fig. 1 in Ref. \cite{kharaz_study_2000}:
\begin{compactitem}
\item ``dataset steel'' (DSS): CoR vs. initial relative velocity for the normal impact of a $5 \; \text{mm}$ diameter aluminum oxide sphere on a thick EN9 steel plate.
\item ``dataset aluminum'' (DSA): CoR vs. initial relative velocity for the normal impact of a $5 \; \text{mm}$ diameter aluminum oxide sphere on a thick aluminum alloy plate.
\end{compactitem}
The data were extracted using the image processing software WebPlotDigitizer \cite{rohatgi_webplotdigitizer_nodate}. In all experiments, a plate was fixed to the ground, and the spheres were dropped from a fixed height, gaining velocity under the influence of the force of gravity on Earth. 

The influence of external forces on the value of the coefficient of restitution was deemed insignificant (see Refs. \cite{quinn_finite_2004, sorace_high_2009, carvalho_exact_2019}). Thus, external forces will be ignored in the model parameter identification study, similarly to how it was done in Ref. \cite{milehins_boucwen_2025}.

The model parameter identification study described in Ref. \cite{milehins_boucwen_2025} was repeated using an implementation of the algorithm COBYQA \cite{ragonneau_model-based_2022, ragonneau_optimal_2024, ragonneau_pdfo_2024} available via the interface of the SciPy function \texttt{optimize.minimize} (in this case, the multi-objective optimization problem described in Sec.  \ref{sec:MPI:Background} reduces to a scalar optimization problem). The approximations of the values of the identified parameters and the associated values of the cost function are shown in Table \ref{tab:DSSDSA_param}. Figure \ref{fig:CoR_NDBWSHCCM} shows the plots of CoR vs. the initial relative velocity obtained experimentally and from the results of the numerical simulations of the NDBWSHCCM. Figure \ref{fig:CoR_NDBWMCM} shows the plots of CoR vs. the initial relative velocity obtained experimentally and from the results of the numerical simulations of the NDBWMCM. As can be seen from the values of the cost function, it was possible to improve the results that were obtained in the previous study, albeit the improvements were marginal. The significant differences between the values of the parameters that were obtained in this study in comparison to the values obtained in the previous study support the claim that the CoR data alone are not sufficient to infer a unique vector of model parameters.

\subsection{Cross (2011)}\label{sec:MPI:Cross}

This subsection presents an update to the model parameter identification study based on the experimentally obtained hysteresis data that was performed in Ref. \cite{milehins_boucwen_2025}. The experimental hysteresis data were provided by Professor Rodney Cross and appeared in Fig. 9.5 in Ref. \cite{cross_physics_2011}. 

Figure \ref{fig:RCross} shows the plots of the experimentally obtained hysteresis loops observed during normal impact of a baseball on a flat surface across a range of initial relative velocities, and the hysteresis loops obtained based on the results of the numerical simulations of the BWSHCCM and the BWMCM with the parameters shown in Table \ref{tab:RCross}. For the purposes of the identification of the parameters of the BWSHCCM and the BWMCM, it was assumed that the only known model parameter was the mass of the ball: its value ($0.145 \; \text{kg}$) was reported in Ref. \cite{cross_physics_2011}. Furthermore, due to the nature of the experiment, it was deemed appropriate to ignore external forces ($u = 0 \; \text{N}$). 

The plots demonstrate a good agreement between the experimentally obtained hysteresis loops and the hysteresis loops obtained from the simulations of the BWSHCCM and the BWMCM. As mentioned in Sec. \ref{sec:contributions}, the previous study \cite{milehins_boucwen_2025} was restricted to the identification of the parameters of the BWSHCCM. The present study shows that the BWMCM can also adequately represent the nature of the physical phenomenon that was described in Ref. \cite{cross_physics_2011}.

\subsection{Villegas et al (2021)}\label{sec:MPI:Quinn}

The final parameter identification study will employ the experimental dataset provided in Fig. 4 in Ref. \cite{villegas_impact_2021}. The figure visualizes the CoR vs. initial relative velocity data and the $t_d$ vs. initial relative velocity data obtained from the measurements of repeated normal impacts of a spring-loaded cart rolling on an inclined surface under the influence of the force of gravity. The data were extracted manually with the assistance of the image processing software WebPlotDigitizer \cite{rohatgi_webplotdigitizer_nodate}. 

The collision model of the experimental setup established by the authors of Ref. \cite{villegas_impact_2021} neglects the forces associated with friction at the contact points of the wheels of the cart with the ground. The same methodology is applied in this article. In this case, both the BWSHCCM and the BWMCM can adequately represent the physical phenomenon, provided that that the external force is given by $u \triangleq - m g \sin \theta$, where $m = 0.506 \; \text{kg}$ is the mass of the cart, $g = 9.78 \; \text{m} \; \text{s}^{-1}$ is the gravitational acceleration (the value of $g$ was reported in Ref. \cite{villegas_impact_2021}), and $\theta=\pi/36 \; \text{rad}$ is the angle of the inclined surface (upon which the cart was rolling) with respect to the ground. The authors of Ref. \cite{villegas_impact_2021} also report the value of the stiffness of the spring attached to the cart: $k = 255 \; \text{kg} \; \text{s}^{-2}$, which was assumed to be linear ($p = 1$). In what follows, for the purposes of model parameter identification, it was assumed that $m$ and $k$ were known while $p$ was allowed to vary. These assumptions lead to the following additional constraints:
\begin{equation}\label{eq:VRBR:c1}
T_b - m U_b = 0
\end{equation}
\begin{equation}\label{eq:VRBR:c2}
\frac{m}{k} \frac{\kappa}{T_b^{p + 1}} - 1 \leq 0
\end{equation}
The constraints follow from the relationships between the dimensional parameters associated with the BWSHCCM and the BWMCM and the nondimensional parameters associated with the NDBWSCHCM and the NDBWMCM, respectively (see Table \ref{tab:ND}, Eq. \eqref{eq:EBWSHCCM_U_nd} and Eq. \eqref{eq:EBWMCM_U_nd}).\footnote{The constraint in Eq. \eqref{eq:VRBR:c1} leads to the reduction of the number of the optimization variables.}

Since both the values of CoR $\tilde{e}$ (or $\tilde{\Theta}_1$) and the duration of the collision $\tilde{t}_d$ (or $\tilde{\Theta}_2$) are reported upon in Ref. \cite{villegas_impact_2021}, the parameter identification problem leads to the bicriteria optimization problem 
\begin{equation}
{\arg\min}_{P^{*}} \; J_j (\tilde{v}_0, \tilde{\Theta}_j, u, P^{*}), \; j \in \{ 1, 2 \}
\end{equation}
subject to the (parameterized) constraints given by Eq. \eqref{eq:op}, Eq. \eqref{eq:VRBR:c1}, and Eq. \eqref{eq:VRBR:c2}, with $\Theta_1$ given by Eq. \eqref{eq:Theta_1} and $\Theta_2$ given by Eq. \eqref{eq:Theta_2}.

It was found empirically (using an implementation of the algorithm COBYQA \cite{ragonneau_model-based_2022, ragonneau_optimal_2024, ragonneau_pdfo_2024} available via the interface of the SciPy function \texttt{optimize.minimize}) that suitably chosen linear scalarizations of the optimization problem can lead to solutions that closely match the experimental data ($J_1, J_2 \sim O(10^{-4})$). The experimental data and the data obtained based on the results of the simulation of the collision models parameterized using a representative vector of identified parameters are shown in Fig. \ref{fig:VRBR_CoR} and Fig. \ref{fig:VRBR_t_d}.

\section{Conclusions and Future Work}\label{sec:conclusions}

The article provided extensions of two mathematical models of binary direct collinear collisions of convex viscoplastic bodies (BWSHCCM and BWMCM) that take into account the effects of external forces. Furthermore, the analysis of the BWMCM was extended to consider certain corner cases that were not considered in the prior study conducted by the authors \cite{milehins_boucwen_2025}.

From the perspective of future work, first of all, it will be useful to establish an unambiguous methodology for the selection of the model parameters based on the properties of the materials and the geometry of the colliding bodies (given the limited amount of experimental contact force data available in the literature, this research will inevitably involve a large-scale experimental study). It will also be useful to extend the modeling framework to other function spaces for the input signals; it will also be beneficial to extend the binary collision model presented in this article to binary collisions of multibody systems or simultaneous collisions of multiple bodies (e.g., see Refs. \cite{brogliato_nonsmooth_2016, stronge_impact_2018}); it may also be useful to consider collision laws developed based on the models of hysteresis other than the Bouc-Wen model (e.g., see Refs. \cite{reid_free_1956, chua_mathematical_1971, chua_generalized_1972, krasnoselskii_systems_1989, mayergoyz_mathematical_1991, macki_mathematical_1993, sadd_contact_1993, visintin_differential_1994, brokate_hysteresis_1996, lacy_hysteretic_2000, oh_modeling_2003, oh_analysis_2003, muravskii_frequency_2004, biswas_reduced-order_2014, biswas_two-state_2015, mielke_rate-independent_2015, biswas_hysteretic_2016, bhattacharjee_interplay_2017, ikhouane_survey_2018, ikhouane_erratum_2018, vaiana_class_2018, vaiana_generalized_2021, vaiana_analytical_2023, vaiana_classification_2023}).

\section*{Acknowledgment} 

The authors would like to acknowledge their families, colleagues, and friends. Special thanks also go to two anonymous reviewers: multiple significant improvements were introduced to the original draft of the article based on their feedback. Special thanks also go to Professor Rodney Cross for providing experimental data from Ref. \cite{cross_physics_2011}. Special thanks also go to the members of staff of Auburn University Libraries for their assistance in finding rare and out-of-print research articles and research monographs. The authors would also like to acknowledge the professional online communities, instructional websites, and various online service providers, especially \url{https://www.adobe.com/acrobat/online/pdf-to-word.html}, \url{https://archive.org/}, \url{https://automeris.io}, \url{https://capitalizemytitle.com}, \url{https://www.matweb.com}, \url{https://www.overleaf.com}, \url{https://pgfplots.net}, \url{https://proofwiki.org/}, \url{https://www.reddit.com}, \url{https://scholar.google.com}, \url{https://stackexchange.com}, \url{https://stringtranslate.com}, \url{https://www.wikipedia.org}. We also note that the results of some of the calculations that are presented in this article were performed with the assistance of the software Wolfram Mathematica \cite{wolfram_research_inc_mathematica_2023}. Other software that was used to produce this article included Adobe Acrobat Reader, Adobe Digital Editions, DiffMerge, Git, GitLab, Google Chrome, Google Gemini (Google Gemini was used as an assistant; no significant parts of the article or the associated code were written by AI), Grammarly (the use of Grammarly was restricted to the identification and correction of spelling, grammar, and punctuation errors), Jupyter Notebook, LibreOffice, macOS Monterey, Mamba, Microsoft Outlook, Preview, Safari, TeX Live/MacTeX, Texmaker, and Zotero.

\section*{Funding Data}

The present work did not receive any specific funding. However, the researchers receive financial support from Auburn University for their overall research activity.

\appendix 

\section{Notation and Conventions}\label{sec:NC}

The notation is adopted from Ref. \cite{milehins_boucwen_2025}, and will not be restated. Essentially all of the definitions and results that are employed in this article are standard in the fields of set theory, general topology, analysis, ordinary differential equations, and nonlinear systems/control. They can be found in a number of textbooks and monographs on these subjects (e.g., see Refs. \cite{takeuti_introduction_1982, kelley_general_2017, khalil_nonlinear_2002, bloch_real_2010, shurman_calculus_2016, ziemer_modern_2017, yoshizawa_stability_1966, yoshizawa_stability_1975, sontag_mathematical_1998, haddad_nonlinear_2011}). Nonetheless, the article employs several concepts that have not appeared in Ref. \cite{milehins_boucwen_2025}. The majority of these concepts are related to the description of dynamics of time-variant systems.

Unless stated otherwise, the time variable for all dynamical systems will be denoted as $t \in \mathbb{R}_{\geq 0}$ (dimensional) or $T \in \mathbb{R}_{\geq 0}$ (nondimensionalized) and $\dot{x}$ will be used to denote the derivative of a differentiable function $x : I \longrightarrow \mathbb{R}^n$ with $I \subseteq \mathbb{R}$ and $n \in \mathbb{Z}_{\geq 1}$. The state variables, inputs, and outputs of a dynamical system may be specified by indicating only their codomains. For example, $q \in \mathbb{R}$ may be used to state that $q$ ranges over the set of real numbers.

\begin{definition}
Consider the following system of ordinary differential equations
\begin{equation}\label{eq:sys}
\dot{x} = f(t, x)
\end{equation}
where $f : \mathbb{R}_{\geq 0} \times \mathbb{R}^n \longrightarrow \mathbb{R}^n$ with $n \in \mathbb{Z}_{\geq 1}$ is the state function that is continuous in the first argument ($t$) and locally Lipschitz continuous in the second argument ($x$). Equation \eqref{eq:sys} augmented with an initial condition $x(t_0) = x_0 \in \mathbb{R}^n$ where $t_0 \in \mathbb{R}_{\geq 0 }$ shall be referred to as an Initial Value Problem (IVP) associated with the system given by Eq. \eqref{eq:sys}. A differentiable function $x : J \longrightarrow \mathbb{R}^n$ with $J \triangleq [t_0, t_0 + T)$ and $T \in \mathbb{R}_{>0} \cup \{ + \infty \}$ is a solution of the IVP associated with the system given by Eq. \eqref{eq:sys} with the initial condition $x_0 \in \mathbb{R}^n$ if $x(t_0) = x_0$ and $\dot{x}(t) = f(t, x(t))$ for all $t \in J$. The system given by Eq. \eqref{eq:sys} may also have an output, which is expressed by the relation $y = g(x)$ or $x \mapsto g(x) \triangleq y$, where $g : \mathbb{R}^n \longrightarrow \mathbb{R}^m$ with $m \in \mathbb{Z}_{\geq 1}$ is a continuous function.\footnote{The output may also depend explicitly on time, but such systems have limited significance in the context of this study.} 
\end{definition}

The following definition was adopted from Ref. \cite{yoshizawa_stability_1975}:\footnote{$\lVert \cdot \rVert$ denotes an arbitrary norm on $\mathbb{R}^n$.}
\begin{definition}
The solutions of the system given by Eq. \eqref{eq:sys} are said to be uniformly bounded if and only if for all $\alpha \in \mathbb{R}_{>0}$, there exists $\beta \in \mathbb{R}_{>0}$ such that $\lVert x(t) \rVert < \beta$ for all $t \in [t_0, t_0 + T)$ for every solution $x : [t_0, t_0 + T) \longrightarrow \mathbb{R}^n$ with $t_0 \in \mathbb{R}_{\geq 0}$ and $T \in \mathbb{R}_{>0} \cup \{ +\infty \}$ starting from the initial condition $x(t_0) = x_0 \in \mathbb{R}^n$ such that $\lVert x_0 \rVert \leq \alpha$.
\end{definition}
The definition can be augmented to consider the outputs of the system:
\begin{definition}
The outputs of the system given by Eq. \eqref{eq:sys} are said to be uniformly bounded if and only if for all $\alpha \in \mathbb{R}_{>0}$, there exists $\gamma \in \mathbb{R}_{>0}$ such that $\lVert g(x(t)) \rVert < \gamma$ for all $t \in [t_0, t_0 + T)$ for every solution $x : [t_0, t_0 + T) \longrightarrow \mathbb{R}^n$ with $t_0 \in \mathbb{R}_{\geq 0}$ and $T \in \mathbb{R}_{>0} \cup \{ +\infty \}$ starting from the initial condition $x(t_0) = x_0 \in \mathbb{R}^n$ such that $\lVert x_0 \rVert \leq \alpha$.
\end{definition}
\begin{prop}\label{thm:ub_solution_imp_ub_output}
If the solutions of the system given by Eq. \eqref{eq:sys} are uniformly bounded, then the outputs of the system given by Eq. \eqref{eq:sys} are uniformly bounded. 
\end{prop}
\begin{proof}
Since $g$ is continuous, the proof follows from the Extreme Value Theorem (e.g., see Theorem 2.4.15 in Ref. \cite{shurman_calculus_2016}).
\end{proof}

\section{Analysis of the BWSHCCM}\label{sec:ABWSHCCM}

Here, the NDBWSHCCM is considered under the assumption that the initial conditions are arbitrary and the values of the parameters are restricted to $B \in \mathbb{R}_{\geq 0}$, $\mathit{\Gamma} \in [-B, B]$, $\kappa \in (0, 1)$, $\sigma \in \mathbb{R}_{\geq 0}$, $n, p \in \mathbb{R}_{\geq 1}$, and $U \in \mathcal{U}_1$. 

Define $\mathcal{W} : \mathbb{R}^3 \longrightarrow \mathbb{R}$ as
\[
\mathcal{W}(\mathbf{X}) \triangleq  \frac{\kappa}{p + 1} \abs{X}^{p + 1} + \frac{\kappa_c}{p + 1} \abs{Z}^{p + 1} + \frac{1}{2} V^2
\]
for all $\mathbf{X} \triangleq (X, Z, V) \in \mathbb{R}^3$.\footnote{The informal notation $\mathbf{X} \triangleq (A_1, \ldots, A_k)$ will be used to introduce a symbol $\mathbf{X}$ for a vector in $\mathbb{R}^k$ with $k \in \mathbb{Z}_{\geq 1}$ and an additional symbol for each of its components.} Define $E : \mathbb{R}_{\geq 0} \longrightarrow \mathbb{R}$ as
\[
E(T) \triangleq \int_0^{T} \abs{U(s)} ds
\]
Note that $E(+\infty) \in \mathbb{R}_{\geq 0}$. Define the Lyapunov function candidate $\mathcal{V} : \mathbb{R}_{\geq 0} \times \mathbb{R}^3 \longrightarrow \mathbb{R}$ as\footnote{This form of the Lyapunov function candidate was inspired by Example 10.1 in Ref. \cite{yoshizawa_stability_1975}.}
\[
\mathcal{V}(T, \mathbf{X}) \triangleq e^{-E(T)} \mathcal{W}(\mathbf{X})
\]
Define $\mathcal{W}' : \mathbb{R}^3 \longrightarrow \mathbb{R}$ as
\[
\mathcal{W}'(\mathbf{X}) \triangleq -\sigma \abs{X}^p V^2 - \kappa_c \abs{Z}^{p + n - 1} \left( B \abs{Z} \abs{V} + \mathit{\Gamma} Z V \right)
\]
for all $\mathbf{X} \in \mathbb{R}^3$. Referring to Ref. \cite{milehins_boucwen_2025}, note that 
\[
\dot{\mathcal{W}}(T, \mathbf{X}) = \mathcal{W}'(\mathbf{X}) + U(T) V
\]
for all $T \in \mathbb{R}_{\geq 0}$ and $\mathbf{X} \in \mathbb{R}^3$. Thus, 
\[
\dot{\mathcal{V}}(T, \mathbf{X}) = e^{- E(T)} \left( U(T) V - \abs{U(T)} \mathcal{W}(\mathbf{X}) \right) + e^{- E(T)} \mathcal{W}'(\mathbf{X})
\]
for all $T \in \mathbb{R}_{\geq 0}$ and $\mathbf{X} \in \mathbb{R}^3$. Define $K \in \mathbb{R}_{\geq 0}$ as
\[
K \triangleq \max \left( \frac{p + 1}{\kappa}, \frac{p + 1}{\kappa_c} \right)
\]

\begin{lem}\label{thm:EBWSHCCM_V_leq_0}
Under the restrictions on the values of the parameters stated above, $\dot{\mathcal{V}}(T, \mathbf{X}) \leq 0$ for all $T \in \mathbb{R}_{\geq 0}$ and $\mathbf{X} \in \mathbb{R}^3$ such that $K \leq \lVert \mathbf{X} \rVert_{\infty}$.\footnote{Note that $\lVert x \rVert_{\infty} \triangleq \max \left( \abs{x_1}, \ldots, \abs{x_n} \right)$ denotes the standard $\infty$-norm on $\mathbb{R}^n$.}
\end{lem}
\begin{proof}

Note that $\mathcal{W}'(\mathbf{X}) \leq 0$ for all $\mathbf{X} \in \mathbb{R}^3$ (a proof can be found in Ref. \cite{milehins_boucwen_2025}). Also, $2 \leq K$ because $(p + 1)/\kappa \leq K$, $\kappa \in (0, 1)$, and $1 \leq p$.

Fix $T \in \mathbb{R}_{\geq 0}$ and $\mathbf{X} \in \mathbb{R}^3$ such that $K \leq \lVert \mathbf{X} \rVert_{\infty}$. Since $\mathcal{W}'(\mathbf{X}) \leq 0$, to show that $\dot{\mathcal{V}}(T, \mathbf{X}) \leq 0$, it suffices to show that $U(T) V \leq \abs{U(T)} \mathcal{W}(\mathbf{X})$. Thus, it suffices to show that $\abs{V} \leq \mathcal{W}(\mathbf{X})$. There are three cases to consider:
\begin{compactitem}
\item Case I: $\lVert \mathbf{X} \rVert_{\infty} = \abs{X}$. It then follows that $(p + 1)/\kappa \leq K \leq \abs{X}$ or $1 \leq \kappa/(p + 1) \abs{X}$, $1 \leq \abs{X}$ and $\abs{V} \leq \abs{X}$. Therefore, 
\[
\abs{V} \leq \abs{X} \leq \abs{X}^{p} \leq \frac{\kappa}{p + 1} \abs{X}^{p + 1} \leq \mathcal{W} (\mathbf{X})
\]
\item Case II: $\lVert \mathbf{X} \rVert_{\infty} = \abs{Z}$. The proof of $\abs{V} \leq \mathcal{W} (\mathbf{X})$ follows from an argument similar to the one used in Case I.
\item Case III: $\lVert \mathbf{X} \rVert_{\infty} = \abs{V}$. Then, $2 \leq K \leq \abs{V}$. Thus, 
\[
\abs{V} \leq \frac{1}{2} V^2 \leq \mathcal{W} (\mathbf{X})
\]
\end{compactitem}
Thus, $\dot{\mathcal{V}}(T, \mathbf{X}) \leq 0$. By generalization, this holds for all $T \in \mathbb{R}_{\geq 0}$ and $\mathbf{X} \in \mathbb{R}^3$ such that $K \leq \lVert \mathbf{X} \rVert_{\infty}$.
\end{proof}

\begin{prop}\label{thm:EBWSHCCM_EUB}
Under the restrictions on the values of the parameters stated above, there exists a unique solution of the NDBWSHCCM on any time interval $[0, T)$ with $T \in \mathbb{R}_{>0} \cup \{ +\infty \}$ for every initial condition $(X_0, Z_0, V_0) \in \mathbb{R}^3$. Furthermore, the solutions of the NDBWSHCCM are uniformly bounded.
\end{prop}
\begin{proof}
Taking into account that the state function of the NDBWSHCCM is locally Lipschitz continuous in $\mathbf{X}$ and continuous in $T$, the solutions of the NDBWSHCCM exist and are unique on a non-empty maximal interval of existence (e.g., see Theorem 54 in Ref. \cite{sontag_mathematical_1998} or Theorem 2.38 in Ref. \cite{haddad_nonlinear_2011}). Noting that $\mathcal{W}$ and $\mathcal{V}$ are continuously differentiable, $\mathcal{W}$ is radially unbounded and positive definite,
\[
e^{-E(+\infty)} \mathcal{W}(\mathbf{X}) \leq \mathcal{V} (T, \mathbf{X}) \leq \mathcal{W}(\mathbf{X})
\]
for all $T \in \mathbb{R}_{\geq 0}$ and $\mathbf{X} \in \mathbb{R}^3$, and $\dot{\mathcal{V}}(T, \mathbf{X}) \leq 0$ for all $T \in \mathbb{R}_{\geq 0}$ and $\mathbf{X} \in \mathbb{R}^3$ such that $K \leq \lVert \mathbf{X} \rVert_{\infty}$ (by Lemma \ref{thm:EBWSHCCM_V_leq_0}), the solutions of the NDBWSHCCM are uniformly bounded by Theorem 8.8 in Ref. \cite{yoshizawa_stability_1975} (see also \cite[Barbashin and Krasovskii (1952), as cited in Ref.][]{kellett_compendium_2014}, Ref. \cite{hahn_stability_1967} and Ref. \cite{kellett_compendium_2014} for a description of a relationship between $\mathcal{K}_{\infty}$-class functions and positive definite radially unbounded functions). Therefore, by the theorem on the extendability of the solutions (e.g., see Proposition C.3.6 in Ref. \cite{sontag_mathematical_1998} or Theorem 2.39 in Ref. \cite{haddad_nonlinear_2011}), each solution can be extended to a unique solution on $[0, +\infty)$.
\end{proof}

\section{Analysis of the BWMCM}\label{sec:ABWMCM}

In what follows, the NDBWMCM will be considered under the assumption that the initial conditions are arbitrary and the values of the parameters are restricted to $B \in \mathbb{R}_{\geq 0}$, $\mathit{\Gamma} \in [-B, B]$, $\kappa \in (0, 1)$, $\sigma \in \mathbb{R}_{> 0}$, $n, p \in \mathbb{R}_{\geq 1}$, and $U \in \mathcal{U}_1$. 

Define $\mathcal{W} : \mathbb{R}^4 \longrightarrow \mathbb{R}$ as
\[
\mathcal{W}(\mathbf{X}) \triangleq (R + \sigma V)^2 + \frac{\kappa}{p + 1} \abs{Y}^{p + 1} + \frac{\kappa_c}{p + 1} \abs{Z}^{p + 1} + \frac{1}{2} V^2
\]
for all $\mathbf{X} \triangleq (R, Y, Z, V) \in \mathbb{R}^4$. Define also $E : \mathbb{R}_{\geq 0} \longrightarrow \mathbb{R}$ as
\[
E(T) \triangleq \int_0^{T} \abs{U(s)} ds
\]
Note that $E(+\infty) \in \mathbb{R}_{\geq 0}$. Define the Lyapunov function candidate $\mathcal{V} : \mathbb{R}_{\geq 0} \times \mathbb{R}^4 \longrightarrow \mathbb{R}$ as\footnote{This form of the Lyapunov function candidate was inspired by Example 10.1 in Ref. \cite{yoshizawa_stability_1975}.}
\[
\mathcal{V}(T, \mathbf{X}) \triangleq e^{-E(T)} \mathcal{W}(\mathbf{X})
\]
Define $\mathcal{W}' : \mathbb{R}^4 \longrightarrow \mathbb{R}$ as
\[
\mathcal{W}'(\mathbf{X}) \triangleq - \frac{1}{\sigma} \dot{R}^2 - \kappa_c \abs{Z}^{p + n - 1} \left( B \abs{Z} \abs{V - \dot{R}} + \mathit{\Gamma} Z (V - \dot{R}) \right)
\]
and $\mathcal{W}'' : \mathbb{R}^4 \longrightarrow \mathbb{R}$ as
\[
\mathcal{W}''(\mathbf{X}) \triangleq (2 \sigma^2 + 1) V + 2 \sigma R
\]
for all $\mathbf{X} \in \mathbb{R}^4$ (it should be remarked that $\dot{R}$ is used as an abbreviation for $\kappa \sigma \abs{Y}^{p - 1} Y + \kappa_c \sigma \abs{Z}^{p - 1} Z$). Referring to Ref. \cite{milehins_boucwen_2025}, note that 
\[
\dot{\mathcal{W}}(T, \mathbf{X}) = \mathcal{W}'(\mathbf{X}) + U(T) \mathcal{W}''(\mathbf{X})
\]
for all $T \in \mathbb{R}_{\geq 0}$ and $\mathbf{X} \in \mathbb{R}^4$. Thus, 
\[
\dot{\mathcal{V}}(T, \mathbf{X}) = e^{- E(T)} \left( U(T) \mathcal{W}''(\mathbf{X}) - \abs{U(T)} \mathcal{W}(\mathbf{X}) \right) + e^{- E(T)} \mathcal{W}'(\mathbf{X})
\]
for all $T \in \mathbb{R}_{\geq 0}$ and $\mathbf{X} \in \mathbb{R}^4$. Define $K \in \mathbb{R}_{\geq 0}$ as
\[
K \triangleq (2 \sigma^2 + 2 \sigma + 1) \max \left( 2 \sigma^2 + 1, \frac{p + 1}{\kappa}, \frac{p + 1}{\kappa_c} \right)
\]

\begin{lem}\label{thm:EBWMCM_V_leq_0}
Under the restrictions on the values of the parameters stated above, $\dot{\mathcal{V}}(T, \mathbf{X}) \leq 0$ for all $T \in \mathbb{R}_{\geq 0}$ and $\mathbf{X} \in \mathbb{R}^4$ such that $K \leq \lVert \mathbf{X} \rVert_{\infty}$.
\end{lem}
\begin{proof}

Note that $\mathcal{W}'(\mathbf{X}) \leq 0$ for all $\mathbf{X} \in \mathbb{R}^4$ (see Lemma C.1 in Ref. \cite{milehins_boucwen_2025} for a proof; it should be remarked that Lemma C.1 in Ref. \cite{milehins_boucwen_2025} holds also under the less restrictive ranges of the parameters that are used in this study). Fix $T \in \mathbb{R}_{\geq 0}$ and $\mathbf{X} \in \mathbb{R}^4$ such that $K \leq \lVert \mathbf{X} \rVert_{\infty}$. Since $\mathcal{W}'(\mathbf{X}) \leq 0$, it suffices to show that \[
U(T) ((2 \sigma^2 + 1) V + 2 \sigma R) \leq \abs{U(T)} \mathcal{W}(\mathbf{X})
\]
Therefore, it suffices to show that 
\[
\abs{(2 \sigma^2 + 1) V + 2 \sigma R} \leq \mathcal{W}(\mathbf{X})
\]
or
\[
(2 \sigma^2 + 1) \abs{V} + 2 \sigma \abs{R} \leq \mathcal{W}(\mathbf{X})
\]
There are four cases to consider:
\begin{compactitem}
\item Case I: $\lVert \mathbf{X} \rVert_{\infty} = \abs{R}$. Then, $\abs{V} \leq \abs{R}$ and
\[
(2 \sigma^2 + 2 \sigma + 1) (2 \sigma^2 + 1) \leq K \leq \abs{R}
\] 
or
\[
\abs{R} \leq (2 \sigma^2 + 2 \sigma + 1)^{-1} (2 \sigma^2 + 1)^{-1} R^2
\]
Therefore, 
\[
\begin{aligned}
(2 \sigma^2 + 1) \abs{V} + 2 \sigma \abs{R} & \leq (2 \sigma^2 + 2 \sigma + 1) \abs{R} \leq \frac{1}{2 \sigma^2 + 1} R^2\\
& \leq (R + \sigma V)^2 + \frac{1}{2} V^2 \leq \mathcal{W}(\mathbf{X})
\end{aligned}
\]
\item Case II: $\lVert \mathbf{X} \rVert_{\infty} = \abs{Y}$. Then, $\abs{R} \leq \abs{Y}$, $\abs{V} \leq \abs{Y}$, and $1 \leq \abs{Y}$. Furthermore,
\[
(2 \sigma^2 + 2 \sigma + 1) \frac{p + 1}{\kappa} \leq K \leq \abs{Y}
\] 
or
\[
\abs{Y} \leq (2 \sigma^2 + 2 \sigma + 1)^{-1} \frac{\kappa}{p + 1} Y^2
\]
Therefore, 
\[
\begin{aligned}
(2 \sigma^2 + 1) \abs{V} + 2 \sigma \abs{R} & \leq (2 \sigma^2 + 2 \sigma + 1) \abs{Y} \leq \frac{\kappa}{p + 1} Y^2\\
& \leq \frac{\kappa}{p + 1} \abs{Y}^{p + 1} \leq \mathcal{W}(\mathbf{X})
\end{aligned}
\]
\item Case III: $\lVert \mathbf{X} \rVert_{\infty} = \abs{Z}$. The proof of $\abs{V} \leq \mathcal{W} (\mathbf{X})$ follows from an argument similar to the one used in Case II.
\item Case IV: $\lVert \mathbf{X} \rVert_{\infty} = \abs{V}$. Then, $\abs{R} \leq \abs{V}$ and $1 \leq \abs{V}$. Furthermore,
\[
2 (2 \sigma^2 + 2 \sigma + 1) \leq K \leq \abs{V}
\] 
or
\[
\abs{V} \leq \frac{1}{2} (2 \sigma^2 + 2 \sigma + 1)^{-1} V^2
\]
Therefore, 
\[
(2 \sigma^2 + 1) \abs{V} + 2 \sigma \abs{R} \leq (2 \sigma^2 + 2 \sigma + 1) \abs{V} \leq \frac{1}{2} V^2 \leq \mathcal{W}(\mathbf{X})
\]
\end{compactitem}
Thus, $\dot{\mathcal{V}}(T, \mathbf{X}) \leq 0$. By generalization, this holds for all $T \in \mathbb{R}_{\geq 0}$ and $\mathbf{X} \in \mathbb{R}^4$ such that $K \leq \lVert \mathbf{X} \rVert_{\infty}$.
\end{proof}

\begin{prop}\label{thm:EBWMCM_EUB}
Under the restrictions on the values of the parameters stated above, there exists a unique solution of the NDBWMCM on any time interval $[0, T)$ with $T \in \mathbb{R}_{>0} \cup \{ +\infty \}$ for every initial condition $(R_0, Y_0, Z_0, V_0) \in \mathbb{R}^4$. Furthermore, the solutions and the outputs of the NDBWMCM are uniformly bounded.
\end{prop}
\begin{proof}
Noting that $\mathcal{W}$ is positive definite and radially unbounded,\footnote{The proof that $\mathcal{W}$ for the NDBWMCM is radially unbounded may not appear to be entirely trivial, but it is still a routine exercise in analysis.} the proof of the global existence, uniqueness, and uniform boundedness of the solutions follows the outline of the proof of Proposition \ref{thm:EBWSHCCM_EUB} by Lemma \ref{thm:EBWMCM_V_leq_0}. The proof of the uniform boundedness of the outputs follows from Proposition \ref{thm:ub_solution_imp_ub_output}.
\end{proof}

\section{Simulation Methodology}\label{sec:SDA}

The numerical simulation and the data analysis that are described in this body work were performed using Python 3.12.11, NumPy 2.3.3 \cite{harris_array_2020}, and SciPy 1.16.2 \cite{virtanen_scipy_2020}, and relied on the IEEE-754 floating point arithmetic (with the default rounding mode) for the quantization of real numbers \cite{ieee_ieee_2019}. All numerical simulations were performed using the explicit adaptive Runge-Kutta method of order 8 (\texttt{DOP853}) \cite{dormand_family_1980, prince_high_1981, hairer_solving_1993} available via the interface of the function \texttt{integrate.solve\_ivp} from the library SciPy 1.16.0 \cite{virtanen_scipy_2020}. All settings of \texttt{integrate.solve\_ivp} were left at their default values, with the exception of the maximum time step (\texttt{max\_step}), the relative tolerance (\texttt{rtol}), and the absolute tolerance (\texttt{atol}): the maximum time step was set to $\approx 10^{-2}$ for the simulation of the NDBWSHCCM and the NDBWMCM and $\approx T_c/100 \; \text{s}$ for the simulation of the BWSHCCM and the BWMCM, the relative tolerance was set to $\approx 10^{-10}$ (for all states) and the absolute tolerance was set to $\approx 10^{-12}$ (for all states). It should be remarked that the numerical ODE solver was chosen due to the relative simplicity of adoption: the authors wanted to prioritize safety over convergence properties. However, due to the non-smoothness of the governing differential equations, better convergence properties may be achieved by using event-capturing techniques and other methods designed for non-smooth systems. The code is available from the personal repository of the corresponding author.\footnote{\url{https://gitlab.com/user9716869/EBWCM}}

\bibliographystyle{asmejour} 

\bibliography{template.bib} 

@article{bouc_modemathematique_1971,
	title = {Modèle {Mathématique} d’{Hystérésis}},
	volume = {24},
	number = {1},
	journal = {Acustica},
	author = {Bouc, R.},
	year = {1971},
	pages = {16--25},
}

@article{bartz_gravity_2023,
	title = {Gravity {Effects} in {Mass}-{Spring}-{Damper} {Models} of {Inelastic} {Collisions}},
	volume = {44},
	doi = {10.1088/1361-6404/acacd7},
	number = {2},
	journal = {European Journal of Physics},
	author = {Bartz, Sean P.},
	year = {2023},
	pages = {025003},
}

@article{sorace_high_2009,
	title = {High {Apparent} {Adhesion} {Energy} in the {Breakdown} of {Normal} {Restitution} for {Binary} {Impacts} of {Small} {Spheres} at {Low} {Speed}},
	volume = {36},
	doi = {10.1016/j.mechrescom.2008.10.009},
	number = {3},
	journal = {Mechanics Research Communications},
	author = {Sorace, C. M. and Louge, M. Y. and Crozier, M. D. and Law, V. H. C.},
	year = {2009},
	pages = {364--368},
}

@article{xiang_comparative_2018,
	title = {A {Comparative} {Study} of the {Dissipative} {Contact} {Force} {Models} for {Collision} {Under} {External} {Spring} {Forces}},
	volume = {13},
	doi = {10.1115/1.4041031},
	number = {10},
	journal = {ASME J Comput Nonlin Dyn},
	author = {Xiang, Dong and Shen, Yinhua and Wei, Yaozhong and You, Mengxing},
	year = {2018},
	pages = {101009},
}

@article{yardeny_experimental_2020,
	title = {Experimental {Investigation} of the {Coefficient} of {Restitution} of {Particles} {Colliding} {With} {Surfaces} in {Air} and {Water}},
	volume = {31},
	doi = {10.1016/j.apt.2020.07.018},
	number = {9},
	journal = {Advanced Powder Technology},
	author = {Yardeny, Idan and Portnikov, Dmitry and Kalman, Haim},
	year = {2020},
	pages = {3747--3759},
}

@article{ye_size-dependent_2017,
	title = {A {Size}-{Dependent} {Viscoelastic} {Normal} {Contact} {Model} for {Particle} {Collision}},
	volume = {106},
	doi = {10.1016/j.ijimpeng.2017.03.020},
	journal = {International Journal of Impact Engineering},
	author = {Ye, Yang and Zeng, Yawu},
	year = {2017},
	pages = {120--132},
}

@article{villegas_impact_2021,
	title = {Impact {Dynamics} for {Gravity}-{Driven} {Motion} of a {Particle}},
	volume = {42},
	doi = {10.1088/1361-6404/abb56c},
	number = {1},
	journal = {European Journal of Physics},
	author = {Villegas, Cesar E. P. and Rojas, Wudmir Y. and Bravo, Carlos and Rocha, Alexandre R.},
	year = {2021},
	pages = {015006},
}

@article{milehins_boucwen_2025,
	title = {The {Bouc}–{Wen} {Model} for {Binary} {Direct} {Collinear} {Collisions} of {Convex} {Viscoplastic} {Bodies}},
	volume = {20},
	doi = {10.1115/1.4068158},
	number = {6},
	journal = {ASME J Comput Nonlin Dyn},
	author = {Milehins, Mihails and Marghitu, Dan B.},
	year = {2025},
	pages = {061005},
}

@article{vaiana_generalized_2021,
	title = {A {Generalized} {Class} of {Uniaxial} {Rate}-{Independent} {Models} for {Simulating} {Asymmetric} {Mechanical} {Hysteresis} {Phenomena}},
	volume = {146},
	doi = {10.1016/j.ymssp.2020.106984},
	journal = {Mechanical Systems and Signal Processing},
	author = {Vaiana, Nicolò and Sessa, Salvatore and Rosati, Luciano},
	year = {2021},
	pages = {106984},
}

@article{vaiana_class_2018,
	title = {A {Class} of {Uniaxial} {Phenomenological} {Models} for {Simulating} {Hysteretic} {Phenomena} in {Rate}-{Independent} {Mechanical} {Systems} and {Materials}},
	volume = {93},
	doi = {10.1007/s11071-018-4282-2},
	number = {3},
	journal = {Nonlinear Dynamics},
	author = {Vaiana, Nicolò and Sessa, Salvatore and Marmo, Francesco and Rosati, Luciano},
	year = {2018},
	pages = {1647--1669},
}

@article{vaiana_classification_2023,
	title = {Classification and {Unified} {Phenomenological} {Modeling} of {Complex} {Uniaxial} {Rate}-{Independent} {Hysteretic} {Responses}},
	volume = {182},
	doi = {10.1016/j.ymssp.2022.109539},
	journal = {Mechanical Systems and Signal Processing},
	author = {Vaiana, Nicolò and Rosati, Luciano},
	year = {2023},
	pages = {109539},
}

@article{vaiana_analytical_2023,
	title = {Analytical and {Differential} {Reformulations} of the {Vaiana}–{Rosati} {Model} for {Complex} {Rate}-{Independent} {Mechanical} {Hysteresis} {Phenomena}},
	volume = {199},
	doi = {10.1016/j.ymssp.2023.110448},
	journal = {Mechanical Systems and Signal Processing},
	author = {Vaiana, Nicolò and Rosati, Luciano},
	year = {2023},
	pages = {110448},
}

@book{hairer_solving_1993,
	address = {Berlin, The Federal Republic of Germany},
	edition = {2},
	series = {Springer {Series} in {Computational} {Mathematics}},
	title = {Solving {Ordinary} {Differential} {Equations} {I}: {Nonstiff} {Problems}},
	volume = {8},
	isbn = {978-3-540-56670-0},
	publisher = {Springer},
	author = {Hairer, Ernst and Nørsett, Syvert P. and Wanner, Gerhard},
	editor = {Bank, R. and Graham, R. L. and Stoer, J. and Varga, R. and Yserentant, H.},
	year = {1993},
}

@book{pfeiffer_multibody_2004,
	address = {Weinheim, The Federal Republic of Germany},
	series = {Wiley {Series} in {Nonlinear} {Science}},
	title = {Multibody {Dynamics} with {Unilateral} {Contacts}},
	isbn = {978-0-471-15565-2},
	publisher = {WILEY-VCH Verlag GmbH \& Co. KGaA},
	author = {Pfeiffer, Friedrich and Glocker, Christoph},
	editor = {Nayfeh, Ali H. and Holden, Arun V.},
	year = {2004},
}

@book{ziemer_modern_2017,
	address = {Cham, The Swiss Confederation},
	edition = {2},
	series = {Graduate {Texts} in {Mathematics}},
	title = {Modern {Real} {Analysis}},
	volume = {278},
	isbn = {978-3-319-64628-2},
	publisher = {Springer International Publishing},
	author = {Ziemer, William P. and Torres, Monica},
	editor = {Axler, Sheldon and Ribet, Kenneth},
	year = {2017},
}

@book{shurman_calculus_2016,
	address = {Cham, The Swiss Confederation},
	series = {Undergraduate {Texts} in {Mathematics}},
	title = {Calculus and {Analysis} in {Euclidean} {Space}},
	isbn = {978-3-319-49314-5},
	publisher = {Springer International Publishing AG},
	author = {Shurman, Jerry},
	editor = {Axler, S. and Ribet, K.},
	year = {2016},
}

@book{morro_mathematical_2023,
	address = {Cham, The Swiss Confederation},
	series = {Modeling and {Simulation} in {Science}, {Engineering} and {Technology}},
	title = {Mathematical {Modelling} of {Continuum} {Physics}},
	isbn = {978-3-031-20814-0},
	publisher = {Springer Nature Switzerland AG},
	author = {Morro, Angelo and Giorgi, Claudio},
	editor = {Bellomo, Nicola and Tezduyar, Tayfun E.},
	year = {2023},
}

@book{stronge_impact_2018,
	address = {Cambridge, The United Kingdom of Great Britain and Northern Ireland},
	edition = {2},
	title = {Impact {Mechanics}},
	isbn = {978-0-521-84188-7},
	publisher = {Cambridge University Press},
	author = {Stronge, William James},
	year = {2018},
}

@book{roithmayr_dynamics_2016,
	address = {New York, NY},
	title = {Dynamics: {Theory} and {Application} of {Kane}’s {Method}},
	isbn = {978-1-107-00569-3},
	publisher = {Cambridge University Press},
	author = {Roithmayr, Carlos M. and Hodges, Dewey H.},
	year = {2016},
}

@book{johnson_contact_1985,
	address = {Cambridge, The United Kingdom of Great Britain and Northern Ireland},
	title = {Contact {Mechanics}},
	isbn = {978-0-521-34796-3},
	publisher = {Cambridge University Press},
	author = {Johnson, K. L.},
	year = {1985},
}

@book{ikhouane_systems_2007,
	address = {Chichester, The United Kingdom of Great Britain and Northern Ireland},
	title = {Systems with {Hysteresis}: {Analysis}, {Identification} and {Control} using the {Bouc}—{Wen} {Model}},
	isbn = {978-0-470-03236-7},
	publisher = {John Wiley \& Sons},
	author = {Ikhouane, Fayçal and Rodellar, José},
	year = {2007},
}

@book{yoshizawa_stability_1975,
	address = {New York, NY},
	series = {Applied {Mathematical} {Sciences}},
	title = {Stability {Theory} and the {Existence} of {Periodic} {Solutions} and {Almost} {Periodic} {Solutions}},
	volume = {14},
	isbn = {978-1-4612-6376-0},
	publisher = {Springer-Verlag New York},
	author = {Yoshizawa, T.},
	year = {1975},
}

@book{hahn_stability_1967,
	address = {New York, NY},
	series = {Die {Grundlehren} de mathematischen {Wissenschaften} in {Einzeldarstellungen}},
	title = {Stability of {Motion}},
	volume = {138},
	isbn = {978-3-642-50085-5},
	publisher = {Springer-Verlag New York},
	author = {Hahn, Wolfgang},
	editor = {Eckmann, B. and van der Waerden, B. L. and Doob, J. L. and Heinz, E. and Hirzebruch, F. and Hopf, E. and Hopf, H. and Maak, W. and Mac Lane, S. and Magnus, W. and Mumford, D. and Postnikov, M. M. and Schmidt, F. K. and Scott, D. S. and Stein, K.},
	translator = {Baartz, Arne P.},
	year = {1967},
}

@article{hunt_coefficient_1975,
	title = {Coefficient of {Restitution} {Interpreted} as {Damping} in {Vibroimpact}},
	volume = {42},
	doi = {10.1115/1.3423596},
	number = {2},
	journal = {ASME J Appl Mech},
	author = {Hunt, K. H. and Crossley, F. R. E.},
	year = {1975},
	pages = {440--445},
}

@article{butcher_characterizing_2000,
	title = {Characterizing {Damping} and {Restitution} in {Compliant} {Impacts} via {Modified} {K}-{V} and {Higher}-{Order} {Linear} {Viscoelastic} {Models}},
	volume = {67},
	doi = {10.1115/1.1308578},
	number = {4},
	journal = {ASME J Appl Mech},
	author = {Butcher, E. A. and Segalman, D. J.},
	year = {2000},
	pages = {831--834},
}

@article{ma_parameter_2004,
	title = {Parameter {Analysis} of the {Differential} {Model} of {Hysteresis}},
	volume = {71},
	doi = {10.1115/1.1668082},
	number = {3},
	journal = {ASME J Appl Mech},
	author = {Ma, F. and Zhang, H. and Bockstedte, A. and Foliente, G. C. and Paevere, P.},
	year = {2004},
	pages = {342--349},
}

@article{harris_array_2020,
	title = {Array {Programming} {With} {NumPy}},
	volume = {585},
	doi = {10.1038/s41586-020-2649-2},
	number = {7825},
	journal = {Nature},
	author = {Harris, Charles R. and Millman, K. Jarrod and van der Walt, Stéfan J. and Gommers, Ralf and Virtanen, Pauli and Cournapeau, David and Wieser, Eric and Taylor, Julian and Berg, Sebastian and Smith, Nathaniel J. and Kern, Robert and Picus, Matti and Hoyer, Stephan and van Kerkwijk, Marten H. and Brett, Matthew and Haldane, Allan and Fernández del Río, Jaime and Wiebe, Mark and Peterson, Pearu and Gérard-Marchant, Pierre and Sheppard, Kevin and Reddy, Tyler and Weckesser, Warren and Abbasi, Hameer and Gohlke, Christoph and Oliphant, Travis E.},
	year = {2020},
	pages = {357--362},
}

@article{nikravesh_determination_2023,
	title = {Determination of {Effective} {Mass} for {Continuous} {Contact} {Models} in {Multibody} {Dynamics}},
	volume = {58},
	doi = {10.1007/s11044-022-09859-4},
	number = {3-4},
	journal = {Multibody System Dynamics},
	author = {Nikravesh, Parviz E. and Poursina, Mohammad},
	year = {2023},
	pages = {253--273},
}

@article{chatterjee_approximate_2022,
	title = {Approximate {Coefficient} of {Restitution} for {Nonlinear} {Viscoelastic} {Contact} {With} {External} {Load}},
	volume = {24},
	doi = {10.1007/s10035-022-01284-w},
	number = {4},
	journal = {Granular Matter},
	author = {Chatterjee, Abhishek and James, Guillaume and Brogliato, Bernard},
	year = {2022},
	pages = {124},
}

@article{virtanen_scipy_2020,
	title = {Scipy 1.0: {Fundamental} {Algorithms} for {Scientific} {Computing} in {Python}},
	volume = {17},
	doi = {10.1038/s41592-019-0686-2},
	number = {3},
	journal = {Nature Methods},
	author = {Virtanen, Pauli and Gommers, Ralf and Oliphant, Travis E. and Haberland, Matt and Reddy, Tyler and Cournapeau, David and Burovski, Evgeni and Peterson, Pearu and Weckesser, Warren and Bright, Jonathan and van der Walt, Stéfan J. and Brett, Matthew and Wilson, Joshua and Millman, K. Jarrod and Mayorov, Nikolay and Nelson, Andrew R. J. and Jones, Eric and Kern, Robert and Larson, Eric and Carey, C J and Polat, Ilhan and Feng, Yu and Moore, Eric W. and VanderPlas, Jake and Laxalde, Denis and Perktold, Josef and Cimrman, Robert and Henriksen, Ian and Quintero, E. A. and Harris, Charles R. and Archibald, Anne M. and Ribeiro, Antônio H. and Pedregosa, Fabian and van Mulbregt, Paul and {SciPy 1.0 Contributors}},
	year = {2020},
	pages = {261--272},
}

@article{corral_nonlinear_2021,
	title = {Nonlinear {Phenomena} of {Contact} in {Multibody} {Systems} {Dynamics}: {A} {Review}},
	volume = {104},
	doi = {10.1007/s11071-021-06344-z},
	number = {2},
	journal = {Nonlinear Dynamics},
	author = {Corral, Eduardo and Moreno, Raúl Gismeros and García, M. J. Gómez and Castejón, Cristina},
	year = {2021},
	pages = {1269--1295},
}

@article{wen_method_1976,
	title = {Method for {Random} {Vibration} of {Hysteretic} {Systems}},
	volume = {102},
	doi = {10.1061/JMCEA3.0002106},
	number = {2},
	journal = {Journal of the Engineering Mechanics Division},
	author = {Wen, Yi-Kwei},
	year = {1976},
	pages = {249--263},
}

@misc{rohatgi_webplotdigitizer_nodate,
	title = {{WebPlotDigitizer}},
	url = {https://automeris.io},
	author = {Rohatgi, Ankit},
}

@book{brokate_hysteresis_1996,
	address = {New York, NY},
	series = {Applied {Mathematical} {Sciences}},
	title = {Hysteresis and {Phase} {Transitions}},
	volume = {121},
	isbn = {978-1-4612-8478-9},
	doi = {10.1007/978-1-4612-4048-8},
	publisher = {Springer-Verlag New York},
	author = {Brokate, Martin and Sprekels, Jürgen},
	editor = {Marsden, J. E. and Sirovich, L. and John, F.},
	year = {1996},
}

@book{visintin_differential_1994,
	address = {Berlin, The Federal Republic of Germany},
	series = {Applied {Mathematical} {Sciences}},
	title = {Differential {Models} of {Hysteresis}},
	volume = {111},
	isbn = {978-3-642-08132-3},
	publisher = {Springer-Verlag},
	author = {Visintin, Augusto},
	editor = {John, F. and Marsden, J. E. and Sirovich, L.},
	year = {1994},
}

@book{krasnoselskii_systems_1989,
	address = {Berlin, The Federal Republic of Germany},
	title = {Systems with {Hysteresis}},
	isbn = {978-3-642-61302-9},
	publisher = {Springer-Verlag},
	author = {Krasnosel’skiǐ, Mark A. and Pokrovskiǐ, Aleksei V.},
	translator = {Niezgódka, Marek},
	year = {1989},
	note = {[Translated by Marek Niezgódka; originally published as Sistemy s Gisteresisom, Nauka, Moscow, The Russian Soviet Federative Socialist Republic, 1983]},
}

@book{mayergoyz_mathematical_1991,
	address = {New York, NY},
	title = {Mathematical {Models} of {Hysteresis}},
	isbn = {978-1-4612-7767-5},
	publisher = {Springer-Verlag New York},
	author = {Mayergoyz, I. D.},
	year = {1991},
}

@misc{wolfram_research_inc_mathematica_2023,
	address = {Champaign, IL},
	title = {Mathematica, {Version} 13.3},
	url = {https://www.wolfram.com/mathematica},
	author = {{Wolfram Research Inc}},
	year = {2023},
}

@book{khalil_nonlinear_2002,
	address = {Upper Saddle River, NJ},
	title = {Nonlinear {Systems}},
	isbn = {978-0-13-228024-2},
	publisher = {Prentice Hall},
	author = {Khalil, Hassan K.},
	year = {2002},
}

@article{prince_high_1981,
	title = {High {Order} {Embedded} {Runge}-{Kutta} {Formulae}},
	volume = {7},
	doi = {10.1016/0771-050X(81)90010-3},
	number = {1},
	journal = {Journal of Computational and Applied Mathematics},
	author = {Prince, P. J. and Dormand, J. R.},
	year = {1981},
	pages = {67--75},
}

@article{dormand_family_1980,
	title = {A {Family} of {Embedded} {Runge}-{Kutta} {Formulae}},
	volume = {6},
	doi = {10.1016/0771-050X(80)90013-3},
	number = {1},
	journal = {Journal of Computational and Applied Mathematics},
	author = {Dormand, J. R. and Prince, P. J.},
	year = {1980},
	pages = {19--26},
}

@article{sadd_contact_1993,
	title = {Contact {Law} {Effects} on {Wave} {Propagation} in {Particulate} {Materials} {Using} {Distinct} {Element} {Modeling}},
	volume = {28},
	doi = {10.1016/0020-7462(93)90061-O},
	number = {2},
	journal = {International Journal of Non-Linear Mechanics},
	author = {Sadd, Martin H. and Tai, QiMing and Shukla, Arun},
	year = {1993},
	pages = {251--265},
}

@article{biswas_reduced-order_2014,
	title = {A {Reduced}-{Order} {Model} {From} {High}-{Dimensional} {Frictional} {Hysteresis}},
	volume = {470},
	doi = {10.1098/rspa.2013.0817},
	number = {2166},
	journal = {Proceedings of the Royal Society A: Mathematical, Physical and Engineering Sciences},
	author = {Biswas, Saurabh and Chatterjee, Anindya},
	year = {2014},
	pages = {20130817},
}

@article{maxwell_dynamical_1867,
	title = {On the {Dynamical} {Theory} of {Gases}},
	volume = {157},
	doi = {10.1098/rstl.1867.0004},
	journal = {Philosophical Transactions of the Royal Society of London},
	author = {Maxwell, James Clerk},
	year = {1867},
	pages = {49--88},
}

@article{kharaz_study_2000,
	title = {A {Study} of the {Restitution} {Coefficient} in {Elastic}-{Plastic} {Impact}},
	volume = {80},
	doi = {10.1080/09500830050110486},
	number = {8},
	journal = {Philosophical Magazine Letters},
	author = {Kharaz, A. H. and Gorham, D. A.},
	year = {2000},
	pages = {549--559},
}

@article{machado_compliant_2012,
	title = {Compliant {Contact} {Force} {Models} in {Multibody} {Dynamics}: {Evolution} of the {Hertz} {Contact} {Theory}},
	volume = {53},
	doi = {10.1016/j.mechmachtheory.2012.02.010},
	journal = {Mechanism and Machine Theory},
	author = {Machado, Margarida and Moreira, Pedro and Flores, Paulo and Lankarani, Hamid M.},
	year = {2012},
	pages = {99--121},
}

@article{kellett_compendium_2014,
	title = {A {Compendium} of {Comparison} {Function} {Results}},
	volume = {26},
	doi = {10.1007/s00498-014-0128-8},
	number = {3},
	journal = {Mathematics of Control, Signals, and Systems},
	author = {Kellett, Christopher M.},
	year = {2014},
	pages = {339--374},
}

@article{biswas_two-state_2015,
	title = {A {Two}-{State} {Hysteresis} {Model} {From} {High}-{Dimensional} {Friction}},
	volume = {2},
	doi = {10.1098/rsos.150188},
	number = {7},
	journal = {Royal Society Open Science},
	author = {Biswas, Saurabh and Chatterjee, Anindya},
	year = {2015},
	pages = {150188},
}

@article{shen_contact_2018,
	title = {A {Contact} {Force} {Model} {Considering} {Constant} {External} {Forces} for {Impact} {Analysis} in {Multibody} {Dynamics}},
	volume = {44},
	doi = {10.1007/s11044-018-09638-0},
	number = {4},
	journal = {Multibody System Dynamics},
	author = {Shen, Yinhua and Xiang, Dong and Wang, Xiang and Jiang, Li and Wei, Yaozhong},
	year = {2018},
	pages = {397--419},
}

@book{takeuti_introduction_1982,
	address = {New York, NY},
	edition = {2},
	series = {Graduate {Texts} in {Mathematics}},
	title = {Introduction to {Axiomatic} {Set} {Theory}},
	volume = {1},
	isbn = {978-1-4613-8168-6},
	publisher = {Springer-Verlag New York},
	author = {Takeuti, Gaisi and Zaring, Wilson M.},
	editor = {Halmos, P. R. and Gehring, F. W. and Moore, C. C.},
	year = {1982},
}

@book{yoshizawa_stability_1966,
	address = {Tokyo, Japan},
	series = {Publications of the {Mathematical} {Society} of {Japan}},
	title = {Stability {Theory} by {Liapunov}'s {Second} {Method}},
	number = {9},
	publisher = {The Mathematical Society of Japan},
	author = {Yoshizawa, T.},
	year = {1966},
}

@misc{ieee_ieee_2019,
	address = {New York, NY},
	title = {{IEEE} {Standard} for {Floating}-{Point} {Arithmetic}. {IEEE} {Std} {754TM}-2019 ({Revision} of {IEEE} {Std} 754-2008)},
	doi = {10.1109/IEEESTD.2019.8766229},
	publisher = {IEEE},
	author = {{IEEE}},
	year = {2019},
}

@book{newton_mathematical_1729,
	address = {London, The Kingdom of Great Britain},
	title = {The {Mathematical} {Principles} of {Natural} {Philosophy}},
	publisher = {Benjamin Motte},
	author = {Newton, Isaac},
	translator = {Motte, Andrew},
	year = {1729},
	note = {[Translated by A. Motte, originally published as Philosophiæ Naturalis Principia Mathematica by Joseph Streater, London, The Kingdom of England, 1687]},
}

@book{kelley_general_2017,
	address = {Mineola, NY},
	title = {General {Topology}},
	publisher = {Dover Publications Inc},
	author = {Kelley, John L.},
	year = {2017},
	note = {[Originally published as General Topology, D. Van Nostrand Company, New York, NY, 1955]},
}

@book{panagiotopoulos_inequality_1985,
	address = {Boston, MA},
	title = {Inequality {Problems} in {Mechanics} and {Applications}: {Convex} and {Nonconvex} {Energy} {Functions}},
	isbn = {978-3-7643-3094-5},
	publisher = {Birkhäuser Boston},
	author = {Panagiotopoulos, P. D.},
	year = {1985},
}

@book{sontag_mathematical_1998,
	address = {New York, NY},
	edition = {2},
	series = {Texts in {Applied} {Mathematics}},
	title = {Mathematical {Control} {Theory}: {Deterministic} {Finite} {Dimensional} {Systems}},
	volume = {6},
	isbn = {978-0-387-98489-6},
	publisher = {Springer Science+Business Media},
	author = {Sontag, Eduardo D.},
	editor = {Marsden, J. E. and Sirovich, L. and Golubitsky, M. and Jäger, W.},
	year = {1998},
}

@book{bloch_real_2010,
	address = {New York, NY},
	title = {The {Real} {Numbers} and {Real} {Analysis}},
	isbn = {978-0-387-72176-7},
	publisher = {Springer Science+Business Media},
	author = {Bloch, Ethan D.},
	year = {2010},
}

@book{cross_physics_2011,
	address = {New York, NY},
	title = {Physics of {Baseball} \& {Softball}},
	isbn = {978-1-4419-8112-7},
	publisher = {Springer Science+Business Media},
	author = {Cross, Rod},
	year = {2011},
}

@book{haddad_nonlinear_2011,
	address = {Princeton, NJ},
	title = {Nonlinear {Dynamical} {Systems} and {Control}: {A} {Lyapunov}-{Based} {Approach}},
	isbn = {978-1-4008-4104-2},
	publisher = {Princeton University Press},
	author = {Haddad, Wassim M. and Chellaboina, VijaySekhar},
	year = {2011},
}

@book{stewart_dynamics_2011,
	address = {Philadelphia, PA},
	title = {Dynamics with {Inequalities}: {Impacts} and {Hard} {Constraints}},
	isbn = {978-1-61197-070-8},
	publisher = {Society for Industrial and Applied Mathematics},
	author = {Stewart, David E.},
	year = {2011},
}

@book{goebel_hybrid_2012,
	address = {Princeton, NJ},
	title = {Hybrid {Dynamical} {Systems}: {Modeling}, {Stability}, and {Robustness}},
	isbn = {978-1-4008-4263-6},
	publisher = {Princeton University Press},
	author = {Goebel, Rafal and Sanfelice, Ricardo G. and Teel, Andrew R.},
	year = {2012},
}

@book{logan_applied_2013,
	address = {Hoboken, NJ},
	edition = {4},
	title = {Applied {Mathematics}},
	isbn = {978-1-118-47580-5},
	publisher = {John Wiley \& Sons},
	author = {Logan, J. David},
	year = {2013},
}

@book{sanfelice_hybrid_2021,
	address = {Princeton, NJ},
	title = {Hybrid {Feedback} {Control}},
	isbn = {978-0-691-18022-9},
	publisher = {Princeton University Press},
	author = {Sanfelice, Ricardo G.},
	year = {2021},
}

@inproceedings{bouc_forced_1968,
	address = {Prague, Czechoslovakia, September 5-9, 1967},
	title = {Forced {Vibration} of {Mechanical} {Systems} with {Hysteresis}},
	booktitle = {Proceedings of the {Fourth} {Conference} on {Nonlinear} {Oscillations}},
	publisher = {Academia Publishing House of the Czechoslovak Academy of Sciences},
	author = {Bouc, R.},
	editor = {Gonda, Ján and Jelínek, František},
	year = {1968},
	pages = {315},
}

@book{brogliato_nonsmooth_2016,
	address = {Cham, The Swiss Confederation},
	edition = {3},
	series = {Communications and {Control} {Engineering}},
	title = {Nonsmooth {Mechanics}: {Models}, {Dynamics} and {Control}},
	isbn = {978-3-319-28664-8},
	publisher = {Springer International Publishing AG Switzerland},
	author = {Brogliato, Bernard},
	editor = {Isidori, Alberto and van Schuppen, Jan H. and Sontag, Eduardo D. and Krstic, Miroslav},
	year = {2016},
}

@article{quinn_finite_2004,
	title = {Finite {Duration} {Impacts} {With} {External} {Forces}},
	volume = {72},
	doi = {10.1115/1.1875552},
	number = {5},
	journal = {ASME J Appl Mech},
	author = {Quinn, D. Dane},
	year = {2004},
	pages = {778--784},
}

@article{tatara_effects_1977,
	title = {Effects of {External} {Force} on {Contacting} {Times} and {Coefficients} of {Restitution} in a {Periodic} {Collision}},
	volume = {44},
	doi = {10.1115/1.3424175},
	number = {4},
	journal = {ASME J Appl Mech},
	author = {Tatara, Y.},
	year = {1977},
	pages = {773--774},
}

@article{ragonneau_pdfo_2024,
	title = {{PDFO}: {A} {Cross}-{Platform} {Package} for {Powell}’s {Derivative}-{Free} {Optimization} {Solvers}},
	volume = {16},
	doi = {10.1007/s12532-024-00257-9},
	number = {4},
	journal = {Mathematical Programming Computation},
	author = {Ragonneau, Tom M. and Zhang, Zaikun},
	year = {2024},
	pages = {535--559},
}

@article{akhan_low_2024,
	title = {Low {Speed} {Impact} of an {Elastic} {Ball} with {Tapes} and {Clay} {Court}},
	volume = {14},
	doi = {10.3390/app14135674},
	number = {13},
	journal = {Applied Sciences},
	author = {Akhan, Ahmet F. and Marghitu, Dan B.},
	year = {2024},
	pages = {5674},
}

@article{falcon_behavior_1998,
	title = {Behavior of {One} {Inelastic} {Ball} {Bouncing} {Repeatedly} off the {Ground}},
	volume = {3},
	doi = {10.1007/s100510050283},
	number = {1},
	journal = {The European Physical Journal B - Condensed Matter and Complex Systems},
	author = {Falcon, E. and Laroche, C. and Fauve, S. and Coste, C.},
	year = {1998},
	pages = {45--57},
}

@article{carvalho_exact_2019,
	title = {Exact {Restitution} and {Generalizations} for the {Hunt}–{Crossley} {Contact} {Model}},
	volume = {139},
	doi = {10.1016/j.mechmachtheory.2019.03.028},
	journal = {Mechanism and Machine Theory},
	author = {Carvalho, André S. and Martins, Jorge M.},
	year = {2019},
	pages = {174--194},
}

@article{ragonneau_optimal_2024,
	title = {An {Optimal} {Interpolation} {Set} for {Model}-{Based} {Derivative}-{Free} {Optimization} {Methods}},
	volume = {39},
	doi = {10.1080/10556788.2024.2330635},
	number = {4},
	journal = {Optimization Methods and Software},
	author = {Ragonneau, Tom M. and Zhang, Zaikun},
	year = {2024},
	pages = {898--910},
}

@article{shen_contact_2024,
	title = {A {Contact} {Force} {Calculation} {Approach} for {Collision} {Analysis} {With} {Zero} or {Non}-zero {Initial} {Relative} {Velocity}},
	volume = {112},
	issn = {1573-269X},
	doi = {10.1007/s11071-024-10065-4},
	number = {22},
	journal = {Nonlinear Dynamics},
	author = {Shen, Yinhua and Xiang, Dong},
	year = {2024},
	pages = {19795--19808},
}

@phdthesis{movahedi-lankarani_canonical_1988,
	address = {Tucson, AZ},
	type = {Ph. {D}. thesis},
	title = {Canonical {Equations} of {Motion} and {Estimation} of {Parameters} in the {Analysis} of {Impact} {Problems}},
	school = {The University of Arizona},
	author = {Movahedi-Lankarani, Hamid},
	year = {1988},
}

@phdthesis{ragonneau_model-based_2022,
	address = {Hong Kong, People's Republic of China},
	type = {Ph. {D}. thesis},
	title = {Model-{Based} {Derivative}-{Free} {Optimization} {Methods} and {Software}},
	url = {http://arxiv.org/abs/2210.12018},
	school = {The Hong Kong Polytechnic University},
	author = {Ragonneau, Tom M.},
	year = {2022},
}

@book{edgeworth_mathematical_1881,
	address = {London, The United Kingdom of Great Britain and Ireland},
	title = {Mathematical {Psychics}; {An} {Essay} on the {Application} of {Mathematics} to the {Moral} {Sciences}},
	publisher = {C. Kegan Paul \& Co.},
	author = {Edgeworth, Francis Ysidro},
	year = {1881},
}

@book{pareto_manual_2014,
	address = {Oxford, The United Kingdom of Great Britain and Northern Ireland},
	title = {Manual of {Political} {Economy}: {A} {Critical} and {Variorum} {Edition}},
	isbn = {978-0-19-960795-2},
	publisher = {Oxford University Press},
	author = {Pareto, Vilfredo},
	editor = {Montesano, Aldo and Zanni, Alberto and Bruni, Luigino and Chipman, John S. and McLure, Michael},
	year = {2014},
	note = {[Edited by Aldo Montesano, Alberto Zanni, Luigino Bruni, John S. Chipman, and Michael McLure; originally published as the Manuale di Economia Politica con una Introduzione alla Scienza Sociale, Società Editrice Libraria, Milan, Kingdom of Italy, 1906]},
}

@article{ikhouane_erratum_2018,
	title = {Erratum to: {A} {Survey} of the {Hysteretic} {Duhem} {Model}},
	volume = {25},
	doi = {10.1007/s11831-017-9235-2},
	number = {4},
	journal = {Archives of Computational Methods in Engineering},
	author = {Ikhouane, Fayçal},
	year = {2018},
	pages = {1129--1129},
}

@article{ikhouane_survey_2018,
	title = {A {Survey} of the {Hysteretic} {Duhem} {Model}},
	volume = {25},
	doi = {10.1007/s11831-017-9218-3},
	number = {4},
	journal = {Archives of Computational Methods in Engineering},
	author = {Ikhouane, Fayçal},
	year = {2018},
	pages = {965--1002},
}

@article{ouyang_absolute_2014,
	title = {Absolute {Stability} {Analysis} of {Linear} {Systems} {With} {Duhem} {Hysteresis} {Operator}},
	volume = {50},
	doi = {10.1016/j.automatica.2014.04.028},
	number = {7},
	journal = {Automatica},
	author = {Ouyang, Ruiyue and Jayawardhana, Bayu},
	year = {2014},
	pages = {1860--1866},
}

@article{jayawardhana_stability_2012,
	title = {Stability of {Systems} {With} the {Duhem} {Hysteresis} {Operator}: {The} {Dissipativity} {Approach}},
	volume = {48},
	doi = {10.1016/j.automatica.2012.06.069},
	number = {10},
	journal = {Automatica},
	author = {Jayawardhana, Bayu and Ouyang, Ruiyue and Andrieu, Vincent},
	year = {2012},
	pages = {2657--2662},
}

@article{oh_modeling_2003,
	title = {Modeling and {Identification} of {Rate}-{Independent} {Hysteresis} {Using} a {Semilinear} {Duhem} {Model}},
	volume = {36},
	doi = {10.1016/S1474-6670(17)34978-9},
	number = {16},
	journal = {IFAC Proceedings Volumes},
	author = {Oh, JinHyoung and Bernstein, Dennis S.},
	year = {2003},
	pages = {1537--1542},
}

@inproceedings{oh_analysis_2003,
	address = {Maui, HI, December 09-12, 2003},
	title = {Analysis of the {Semilinear} {Duhem} {Model} for {Rate}-{Independent} {Hysteresis}},
	volume = {6},
	doi = {10.1109/CDC.2003.1272285},
	booktitle = {42nd {IEEE} {International} {Conference} on {Decision} and {Control}},
	author = {Oh, JinHyoung and Bernstein, D. S.},
	year = {2003},
	pages = {6236--6241},
}

@inproceedings{lacy_hysteretic_2000,
	address = {Chicago, IL, June 28-30, 2000},
	title = {Hysteretic {Systems} and {Step}-{Convergent} {Semistability}},
	volume = {6},
	doi = {10.1109/ACC.2000.877000},
	booktitle = {Proceedings of the 2000 {American} {Control} {Conference}},
	author = {Lacy, S. L. and Bernstein, D. S. and Bhat, S. P.},
	year = {2000},
	pages = {4139--4143},
}

@article{chua_generalized_1972,
	title = {A {Generalized} {Hysteresis} {Model}},
	volume = {19},
	doi = {10.1109/TCT.1972.1083416},
	number = {1},
	journal = {IEEE Transactions on Circuit Theory},
	author = {Chua, L. O. and Bass, S. C.},
	year = {1972},
	pages = {36--48},
}

@article{chua_mathematical_1971,
	title = {Mathematical {Model} for {Dynamic} {Hysteresis} {Loops}},
	volume = {9},
	doi = {10.1016/0020-7225(71)90046-2},
	number = {5},
	journal = {International Journal of Engineering Science},
	author = {Chua, Leon O. and Stromsmoe, Keith A.},
	year = {1971},
	pages = {435--450},
}

@book{mielke_rate-independent_2015,
	address = {New York, NY},
	series = {Applied {Mathematical} {Sciences}},
	title = {Rate-{Independent} {Systems}: {Theory} and {Application}},
	volume = {193},
	isbn = {978-1-4939-2706-7},
	publisher = {Springer Science+Business Media LLC New York},
	author = {Mielke, Alexander and Roubíček, Tomáš},
	year = {2015},
}

@article{bhattacharjee_interplay_2017,
	title = {Interplay {Between} {Dissipation} and {Modal} {Truncation} in {Ball}-{Beam} {Impact}},
	volume = {12},
	doi = {10.1115/1.4036830},
	number = {6},
	journal = {ASME J Comput Nonlin Dyn},
	author = {Bhattacharjee, Arindam and Chatterjee, Anindya},
	year = {2017},
	pages = {061018},
}

@article{biswas_hysteretic_2016,
	title = {Hysteretic {Damping} in an {Elastic} {Body} {With} {Frictional} {Microcracks}},
	volume = {108-109},
	doi = {10.1016/j.ijmecsci.2016.01.029},
	journal = {International Journal of Mechanical Sciences},
	author = {Biswas, Saurabh and Jana, Prasun and Chatterjee, Anindya},
	year = {2016},
	pages = {61--71},
}

@article{muravskii_frequency_2004,
	title = {On {Frequency} {Independent} {Damping}},
	volume = {274},
	doi = {10.1016/j.jsv.2003.05.012},
	number = {3-5},
	journal = {Journal of Sound and Vibration},
	author = {Muravskii, G. B.},
	year = {2004},
	pages = {653--668},
}

@article{reid_free_1956,
	title = {Free {Vibration} and {Hysteretic} {Damping}},
	volume = {60},
	doi = {10.1017/S0368393100135242},
	number = {544},
	journal = {The Aeronautical Journal},
	author = {Reid, T. J.},
	year = {1956},
	pages = {283--283},
}

@article{euler_force_1746,
	title = {De la {Force} de {Percussion} et de {Sa} {Véritable} {Mesure}},
	volume = {1},
	url = {https://scholarlycommons.pacific.edu/euler-works/82},
	journal = {Mémoires de l'Académie des Sciences de Berlin},
	author = {Euler, Leonhard},
	year = {1746},
	note = {[The article was downloaded from Euler Archive - All Works, 82]},
	pages = {21--53},
}

@article{moore_collision_1988,
	title = {Collision {Detection} and {Response} for {Computer} {Animation}},
	volume = {22},
	doi = {10.1145/378456.378528},
	number = {4},
	journal = {ACM SIGGRAPH Computer Graphics},
	author = {Moore, Matthew and Wilhelms, Jane},
	year = {1988},
	pages = {289--298},
}

@article{platt_constraint_1988,
	title = {Constraint {Methods} for {Flexible} {Models}},
	volume = {22},
	doi = {10.1145/378456.378524},
	number = {4},
	journal = {ACM SIGGRAPH Computer Graphics},
	author = {Platt, John C. and Barr, Alan H.},
	year = {1988},
	pages = {279--288},
}

@article{terzopoulos_elastically_1987,
	title = {Elastically {Deformable} {Models}},
	volume = {21},
	doi = {10.1145/37402.37427},
	number = {4},
	journal = {ACM SIGGRAPH Computer Graphics},
	author = {Terzopoulos, Demetri and Platt, John and Barr, Alan and Fleischer, Kurt},
	year = {1987},
	pages = {205--214},
}

@article{macki_mathematical_1993,
	title = {Mathematical {Models} for {Hysteresis}},
	volume = {35},
	doi = {10.1137/1035005},
	number = {1},
	journal = {SIAM Review},
	publisher = {Society for Industrial and Applied Mathematics},
	author = {Macki, Jack W. and Nistri, Paolo and Zecca, Pietro},
	year = {1993},
	pages = {94--123},
}

@incollection{moreau_unilateral_1988,
	address = {Vienna, The Republic of Austria},
	title = {Unilateral {Contact} and {Dry} {Friction} in {Finite} {Freedom} {Dynamics}},
	isbn = {978-3-7091-2624-0},
	doi = {10.1007/978-3-7091-2624-0_1},
	booktitle = {Nonsmooth {Mechanics} and {Applications}},
	publisher = {Springer-Verlag Wien},
	author = {Moreau, J. J.},
	editor = {Moreau, J. J. and Panagiotopoulos, P. D.},
	year = {1988},
	pages = {1--82},
}

\end{document}